\pdfoutput= 1 
\documentclass{article}
\usepackage{amsmath,amssymb,amsfonts,amsthm,mathtools,xcolor,multicol,microtype}
\usepackage[hidelinks,pdfusetitle]{hyperref}
\usepackage[numbered]{bookmark}
\usepackage[only,llbracket,rrbracket]{stmaryrd}
\newcommand{\intset}[2]{\llbracket #1, #2\rrbracket}

\newtheorem{theorem}{Theorem}
\numberwithin{theorem}{section}
\newtheorem{proposition}[theorem]{Proposition}
\newtheorem{lemma}[theorem]{Lemma}
\newtheorem{definition}[theorem]{Definition}
\newtheorem{conjecture}[theorem]{Conjecture}
\theoremstyle{remark}
\newtheorem{remark}[theorem]{Remark}

\numberwithin{equation}{section}

\newcommand{\unquad}{\hspace{-1em}}

\DeclareMathOperator{\Card}{Card}
\newcommand{\Nsusy}{\mathcal{N}}
\newcommand{\CC}{\mathbb{C}}
\newcommand{\RR}{\mathbb{R}}
\newcommand{\QQ}{\mathbb{Q}}
\newcommand{\ZZ}{\mathbb{Z}}
\newcommand{\qinst}{\mathfrak{q}}
\newcommand{\eps}{\epsilon}
\newcommand{\veps}{\varepsilon}
\newcommand{\bel}{\be\label}
\newcommand{\be}{\begin{equation}}\newcommand{\ee}{\end{equation}}
\newcommand\beal{}
\def\beal#1\ee{\begin{equation}\begin{aligned}\relax #1\end{aligned}\end{equation}}
\newcommand\beall{}
\def\beall#1#2\ee{\begin{equation}\label{#1}\begin{aligned}\relax #2\end{aligned}\end{equation}}
\newcommand\bega{}
\def\bega#1\ee{\begin{equation}\begin{gathered}\relax#1\end{gathered}\end{equation}}
\newcommand{\Zinst}{Z_{\textnormal{inst}}}
\newcommand{\Zinststar}{Z_{\textnormal{inst}}^{\Nsusy=2^*}}
\newcommand{\onebf}{\mathbf{1}}
\DeclareMathOperator{\dist}{dist}
\DeclareMathOperator{\sign}{sgn}
\newcommand{\ceil}[1]{\lceil#1\rceil}
\newcommand{\bigceil}[1]{\bigl\lceil#1\bigr\rceil}
\newcommand{\Bigceil}[1]{\Bigl\lceil#1\Bigr\rceil}

\newcommand{\floor}[1]{\lfloor#1\rfloor}
\newcommand{\bigfloor}[1]{\bigl\lfloor#1\bigr\rfloor}
\newcommand{\Bigfloor}[1]{\Bigl\lfloor#1\Bigr\rfloor}

\newcommand{\abs}[1]{\lvert #1\rvert}
\newcommand{\Rabs}{R_{\textnormal{abs}}}

\begin{document}

\title{Convergence of Nekrasov instanton sum with adjoint matter}
\author{Bruno Le Floch\thanks{Email: {\tt blefloch@lpthe.jussieu.fr}. Centre National de la Recherche Scientifique \& Laboratoire de Physique Théorique et Hautes Énergies, Sorbonne Université.}}
\date{May 2026}

\maketitle

\begin{abstract}
  The Nekrasov instanton partition function of the 4d $\mathcal{N}=2^*$ $U(N)$ gauge theory (a mass deformation of 4d $\mathcal{N}=4$ super-Yang-Mills theory), which is a generating series of equivariant integrals over instanton moduli spaces, is given by a sum over colored partitions weighted by a counting parameter~$\qinst$.
  This note proves convergence of the series in the unit disk $|\qinst|<1$ for generic parameters.
  Specifically, the absolute convergence radius of this sum is determined, assuming that mass and Coulomb branch parameters avoid some lattice.
  If the ratio $b^2=\epsilon_1/\epsilon_2$ of equivariant parameters is in ${\CC\setminus[0,+\infty)}$, the radius is~$1$, as expected.
  If $b^2$ is non-negative, three cases arise:
  the radius is finite if $b^2$ has finite exponential type (a generalization of Brjuno numbers), namely there exists $C>0$ such that $|b^2-p/q|>\exp(-Cq)$ for all integers $p,q\neq 0$;
  the series diverges if $b^2$ is super-exponentially well approximable by rationals;
  and if $b^2$ is rational some terms are singular.
  The AGT correspondence translates these results to convergence of torus one-point conformal blocks of the Virasoro and $W_N$ algebras with non-real~$b$, within the unit disk.  For the Virasoro algebra this corresponds to a central charge in $\CC\setminus[25,+\infty)$.
\end{abstract}

\vfill

\begingroup
\setlength{\columnsep}{30pt}
\begin{multicols}{2}
\scriptsize
\tableofcontents
\end{multicols}
\endgroup

\section{Introduction and main result}

\subsection{Instanton partition functions}

Nekrasov defined and evaluated~\cite{hep-th/0206161} the partition function of large classes of 4d $\Nsusy=2$ Lagrangian theories on the Omega background~$\RR^4_{\eps_1,\eps_2}$, as a means to derive the Seiberg--Witten prepotential in the $\eps_1,\eps_2\to 0$ limit.
He expressed their non-perturbative part\footnote{We shall focus on the instanton part of the Nekrasov partition function, and ignore the classical and one-loop prefactors which have no effect on convergence of the series.} as a series in powers of an instanton counting parameter~$\qinst$,\footnote{When there are multiple gauge groups each one has a counting parameter~$\qinst$ and a sum over~$k$.  In (mass-deformed) SCFT, they take the form $\qinst=e^{2\pi i\tau}$ where $\tau=\frac{4\pi i}{g^2}+\frac{\theta}{2\pi}$ is a complexified gauge coupling.}
\bel{Zinst-orig}
\Zinst(m,a; \qinst) = \sum_{k\geq 0} \qinst^k Z_k(m,a) ,
\ee
which depends on masses~$m$ and Coulomb branch parameters~$a$ in addition to $\eps_1,\eps_2,\qinst$.
See~\cite{1412.7121} for a review.
These partition functions play a central role in mathematical physics and string theory due to their relation with refined topological strings~\cite{hep-th/0701156}, 2d CFT conformal blocks~\cite{0906.3219}, quantum integrable systems~\cite{0908.4052}, isomonodromic deformations~\cite{1207.0787} and black hole quasinormal modes~\cite{1404.5188,2006.06111}.

Commonly, series expansions of quantum field theory observables are merely asymptotic series, with vanishing convergence radius.  The convergence of the series~\eqref{Zinst-orig} is thus far from obvious.
When the 4d $\Nsusy=2$ theory is of class~S, the Nekrasov partition function is related by the AGT correspondence~\cite{0906.3219} to Virasoro and $W_N$~conformal blocks with coupling constant
\be
b=\sqrt{\eps_1/\eps_2} .
\ee
The expectation from this 2d chiral CFT description is that the series converges, with a radius of convergence that only depends on the gauge theory, and not on $\eps_1,\eps_2,m,a$.
This radius is expected to be infinite for asymptotically free theories, and equal to~$1$ in (mass-deformed) SCFTs.

We aim to establish this optimal radius of convergence for one of the prototypical 4d $\Nsusy=2$ theories, the so-called 4d $\Nsusy=2^*$ $U(N)$ theory, which consists of a $U(N)$ vector multiplet coupled to an adjoint hypermultiplet of mass~$m$.
Equivariant localization on the instanton moduli spaces (see e.g.,~\cite{math/0609841} for a mathematical description) provides an explicit formula as a contour integral, which can be evaluated as a sum of residues.
This results in a sum over $N$-tuples of Young diagrams $\vec{Y} = (Y_1,\dots,Y_N)$
\begin{subequations}\begin{align}
\label{Zinst-Young-k}
\Zinst(m,a; \qinst) & = \sum_{k\geq 0} \qinst^k \sum_{|\vec{Y}|=k} Z_{\vec{Y}}(m,a) \\
\label{Zinst-Young-all}
               & = \sum_{\vec{Y}} \qinst^{|\vec{Y}|} Z_{\vec{Y}}(m,a),
\end{align}\end{subequations}%
where $|\vec{Y}|=|Y_1|+\dots+|Y_N|=k$ is the total number of boxes.
The coefficients $Z_{\vec{Y}}(m,a)$ depend on the mass $m\in\CC$ and Coulomb branch parameters $a=(a_1,\dots,a_N)\in\CC^N$ of the $U(N)$ gauge theory, which are equivariant parameters in the moduli space description.
They are products (over boxes of the Young diagrams $Y_I$, $I=1,\dots,N$) of simple rational functions of $m$ and~$a$.
See~\eqref{Zstar-21} for the explicit formula, which we take as a definition of the instanton partition function.

\subsection{Absolute convergence radius}

We are interested in the validity of the expression~\eqref{Zinst-Young-all} as a sum over all tuples of partitions, and especially its \emph{absolute convergence radius}, namely:
\bel{absolute-sum}
\text{What is the convergence radius $\Rabs$ of } \sum_{\vec{Y}} |\qinst|^{|\vec{Y}|} |Z_{\vec{Y}}(m,a)|<+\infty \text{ ?}
\ee
We dub this sum the absolute Nekrasov sum.
Its convergence for $|\qinst|<\Rabs$ implies the convergence of~$\Zinst$ in that disk, but the converse does not hold.

The main outcome of this work is to obtain the \emph{optimal} absolute convergence radius $\Rabs=1$ for generic parameters $\eps_1,\eps_2,a,m$, and lower/upper bounds on~$\Rabs$ when the ratio $b^2=\eps_1/\eps_2$ is non-generic.
We assume throughout that
\be
\eps_1,\eps_2\neq 0 , \qquad
b^2 = \eps_1 / \eps_2 \in \CC\setminus\{0\} ,
\ee
namely we exclude the Nekrasov--Shatashvili limit~\cite{0908.4052} $\eps_2\to 0$ or $\eps_1\to 0$.
The mass~$m$ of the adjoint hypermultiplet, and the Coulomb branch parameters $(a_1,\dots,a_N)\in\CC^N$ of the $U(N)$ gauge theory are assumed to be suitably generic.

\begin{theorem}[On the 4d $\Nsusy=2^*$ instanton partition function]\label{thm:2star}
  Assume that none of the differences $a_I-a_J$ for $1\leq I<J\leq N$, nor the adjoint mass~$m$, are in the closure\footnote{The lattice $\eps_1\ZZ+\eps_2\ZZ$ is discrete if $b^2\in\CC\setminus\RR$ or $b^2\in\QQ$, but is dense in the line $\eps_1\RR=\eps_2\RR$ if $b^2\in\RR\setminus\QQ$.} of the lattice $\eps_1\ZZ+\eps_2\ZZ$.
  For rational $b^2<0$, assume further\footnote{This condition is only used to prove $\Rabs\leq 1$ in \autoref{sec:Rabs_le_1} and is likely unnecessary.} that none of the $a_I-a_J+m$ are in that lattice.
  Then the following holds.
  \begin{itemize}
  \item For $b^2\in\CC\setminus[0,+\infty)$, the absolute convergence radius is $\Rabs=1$.
  \item For irrational $b^2>0$, let its exponential type $B_{\sup}(b^2)\in[0,+\infty]$ be the infimum of $B>0$ such that $\dist(nb^2,\ZZ)\gtrsim e^{-Bn}$ as $n\to+\infty$.  Then
    \bel{Rabs-lower-bound-thm2star}
    A_1 e^{-8B_{\sup}(b^2)/\max(1,b^2)} \leq \Rabs \leq \min(1, A_2 e^{-B_{\sup}(b^2)/(1+b^2)})
    \ee
    where $A_1,A_2>0$ are continuous in $m,a,\eps_1,\eps_2$.
    In particular, $\Rabs>0$ if the exponential type is finite, and $\Rabs=0$ if the exponential type is infinite, in which case the absolute Nekrasov sum diverges for all $\qinst\neq 0$.
  \item For rational $b^2>0$ the sum is meaningless as some terms are singular.
  \end{itemize}
\end{theorem}

\paragraph{On absolute and conditional convergence.}

From its relation to conformal blocks, the instanton partition function is expected to converge in the unit disk $|\qinst|<1$.
For cases where we show $\Rabs<1$, this requires that the $k$-instanton contribution (sum over tuples $\vec{Y}$ with $k$~boxes) features cancellations such that the resulting series in~$\qinst$ has radius of convergence~$1$.\footnote{Note that this phenomenon whereby the absolute convergence radius~$\Rabs$ is less than the conditional convergence radius does \emph{not} occur in a simple series of the form $\sum_{k\geq 0} a_k \qinst^k$, as is well-known from the theory of holomorphic functions.}

A particularly visible instance of this is that if $b^2$ is a positive rational number, or if any of the differences $a_i-a_j$ for $1\leq i,j\leq N$ belongs to the lattice $\eps_1\ZZ+\eps_2\ZZ$, then some terms $Z_{\vec{Y}}(m,a)$ in the sum are infinite hence the expression~\eqref{Zinst-Young-all} is meaningless.  These singularities are removable singularities in the $k$-instanton contribution as can be seen from its contour integral representation.\footnote{Rational approximations of $b^2>0$ are the main obstruction to proving $\Rabs=1$ for all~$b^2$.  It seems worthwhile to analyze how to group terms to cancel poles in the rational $b^2>0$ case.}

The advantage of having an absolutely convergent sum is that it allows any reorganization of terms in~\eqref{Zinst-Young-all}.
This reordering freedom was used implicitly in works that consider limits where instanton partition functions simplify, such as Higgsing limits where it reduces to instanton--vortex partition functions, as they often involving resumming subsets of terms that are not organized simply by their number of boxes.

\paragraph{Convergence techniques.}

The present work relies on bounding each term in the sum over partitions.
This follows and refines the approach in~\cite{1403.1235}\footnote{I thank Fabrizio Del Monte for useful discussions.} which showed convergence for $b^2=-1$ for pure $\Nsusy=2$ super-Yang--Mills theory (see also~\cite{1608.02566} in 5d), and \cite{2212.06741} where it was shown that $\Zinst$ has a positive radius of convergence for $\Nsusy=2^*$ (excluding the case $b^2\geq 0$) and for $N_f=2N$ SQCD (for $b^2=-1$), and infinite radius of convergence in related asymptotically-free theories.
The relevant combinatorial tools to improve the radius and extend the range of~$b^2$ are rather different for the two theories; the refined analysis of SQCD is thus postponed.
For the $\Nsusy=2^*$ theory in the regime $b^2\in\CC\setminus[0,+\infty)$, \cite{2212.06741}~gave
\bel{Arnaudo-Rabs}
\Rabs \geq \Bigl(1 + \frac{|m|}{\min_{1\leq I\neq J\leq N} \dist(a_I-a_J,\eps_1\ZZ+\eps_2\ZZ)}\Bigr)^{-2(N-1)} ,
\ee
by bounding most factors in~$Z_{\vec{Y}}$ uniformly.
Theorem~\ref{thm:2star} improves this to the \emph{optimal} convergence radius
\be
\Rabs=1
\ee
by observing that only a small number of factors can be as large as the bound used to derive~\eqref{Arnaudo-Rabs}, and the remaining factors tend to~$1$ at a controlled rate.
Roughly speaking, ``box contents'' tend to infinity at a rate~$|\vec{Y}|^{1/2}$, the corresponding factors are $1+O(|\vec{Y}|^{-1/2})=\exp(O(|\vec{Y}|^{-1/2}))$, so that their product is $\exp(O(|\vec{Y}|^{1/2}))$, which is not large enough to affect the convergence radius.
The technique appears to extend straightforwardly to the 5d setting.

The convergence of instanton partition functions (equivalently conformal blocks) and their 5d $\Nsusy=1$ analogues, has been explored in several other ways.
\begin{itemize}
\item Bounding the contour integral representation of 5d $k$-instanton partition functions yields a positive radius of convergence~\cite{1709.05232} analogous to~\eqref{Arnaudo-Rabs} for a range of parameters.  The technique encounters obstructions for 4d partition functions due to infinite integration contours.  To translate the improvements we make here in terms of the contour integral, one would have to track how integration variables ``spread out'' away from each other, in analogy to how the box content grows to infinity.

\item A large body of mathematical work on vertex operator algebras~(VOA) established convergence of their conformal blocks with arbitrary genus, assuming that the VOA is $C_2$-cofinite and the VOA modules involved are $C_1$-cofinite.  See \cite{2204.04409} and references therein.  In the present setting, this corresponds to conformal blocks of minimal models and $W_N$~minimal models, namely rational $b^2<0$ and a finite set of values for $m,a$.  Convergence is proven by relying on BPZ equations arising from the fact that the theory is a rational CFT\@.

\item The probabilistic construction of Liouville CFT from Gaussian multiplicative chaos allowed to prove convergence of conformal blocks for $b^2\in(0,1)$ and $m/\eps_2\in(-1,b^2+2)$ and purely imaginary $a/\eps_2$ in~\cite{2003.03802}.
\end{itemize}

\paragraph{Other contexts where Diophantine properties of $b^2$ matter.}

The limit of minimal model correlators (negative rational~$b^2$) as $b^2$ tends to a negative irrational value was considered in \cite{1406.4290,1711.08916,1909.10784}.  When considering the non-diagonal D-series minimal models, this led to a non-diagonal irrational CFT sharing some similarities with Liouville CFT, provided the irrational value of~$b^2$ is not too well approximated by rationals.  The precise convergence criterion has not been fully elucidated and it would be interesting to compare with the present work.

Let us also mention a recent series representation of Virasoro fusion (and modular) kernels~\cite{2405.09325,2512.03172}\footnote{Another recent work~\cite{2508.14030} provides a closed form expression of these kernels for $c=1$, namely $b^2=-1$.} in powers of $\exp(-2\pi b^{\pm 1}P_t)$ with $P_t$ one of the internal momenta.  Its convergence properties for $b^2>0$ were conjectured to depend on arithmetic properties of~$b^2$, in analogy with the convergence of basic hypergeometric series~\cite{1504.01238}.  It would be interesting to determine whether the criteria we develop for the absolute Nekrasov sum are relevant for these series.

\subsection{Conformal blocks and more}
\label{sec:conf-blocks}

The AGT dictionary (see~\cite{2006.14025} and references therein) relates the instanton partition function of interest here to torus one-point conformal blocks for a product of the Heisenberg algebra and the $W_N$~algebra (corresponding to the $\mathfrak{u}(1)\times\mathfrak{su}(N)$ decomposition of $\mathfrak{u}(N)$).
The Heisenberg algebra contribution can be factored out, leaving a conformal block for the $W_N$ algebra with central charge
\be
c_{\mathfrak{su}(N)} = N-1 + (N^3-N) (b+1/b)^2 ,
\ee
and where the external vertex operator inserted on the torus is semi-degenerate, hence depends only on a single parameter $m\in\CC$, while the internal operator flowing through the torus (VOA module on which the trace is taken) depends on the $N$ Coulomb branch parameters modulo overall translation.
In the case $N=2$, the semi-degeneracy condition trivializes and one obtains arbitrary one-point torus conformal blocks of the Virasoro algebra.

As a result, Theorem~\ref{thm:2star} establishes the expected convergence region $|\qinst|<1$ of these conformal blocks in the regime
\be
c_{\mathfrak{su}(N)} \in \CC \setminus \bigl[ (N-1)(1+2N)^2 , +\infty\bigr) , \qquad
\text{generic } m, a .
\ee
The genericity condition on~$m$ states that it is not a degenerate vertex operator.
The genericity condition on~$a$ is less readily interpreted.

In the Virasoro case this means convergence is established for
\be
c_{\textnormal{Vir}}\in \CC\setminus[25,+\infty) , \qquad
\text{generic } m, a .
\ee
This includes in particular the case of minimal model central charges corresponding to $b^2$ a negative rational number (with $c<1$), but the genericity condition on~$m$ \emph{excludes} one-point blocks for the operators that are in the spectrum of minimal models.
Relaxing the genericity conditions on $m,a$ to reproduce convergence of minimal model conformal blocks seems quite difficult: indeed, as explored in \cite{1404.7075,1404.7094,1504.01925,1507.03540,1509.01000} one may need to restrict the instanton sum to only certain tuples of partitions.

Nekrasov partition sums also appear in the study of isomonodromic deformations.  Specifically, isomonodromic tau functions take the form of Fredholm determinants of suitable operators, which can be expanded as series over tuples of Young diagrams that match Nekrasov functions.  This was understood in \cite{1608.00958,1705.01869,1712.08546} for 4d $\Nsusy=2$ linear quiver theories including $SU(N)$ SQCD, and in~\cite{2011.06292} for circular quiver gauge theories, which includes the $\Nsusy=2^*$ theory.  Convergence results in the present paper should also apply to this setting.

\section{Strategy and basic upper bound}

\subsection{Formulas and notation}

The instanton partition function for the $\Nsusy=2^*$ theory takes the form (see \autoref{sec:instanton-telescopic})
\beall{Zstar-21}
\Zinststar & = \sum_{\vec{Y}} \qinst^{|\vec{Y}|} Z_{\vec{Y}} , \qquad Z_{\vec{Y}} = \prod_{I,J=1}^N Z_{IJ} , \\
Z_{IJ} & = Z_{Y_IY_J}(a_I-a_J+\eps_2,m) Z_{Y_IY_J}(a_I-a_J-\eps_1,-m) , \\
Z_{\lambda\mu}(a,m) & = \prod_{(i,j)\in\lambda} \Bigl(1-\frac{m}{a-\eps_1 (\mu'_j-i) + \eps_2 (\lambda_i-j)}\Bigr) .
\ee

As for standard analytic series, the absolute convergence radius~$\Rabs$ can be expressed as
\bel{Rdef-Z}
\Rabs = \liminf_{\vec{Y}} |Z_{\vec{Y}}|^{-1/|\vec{Y}|} \in [0,+\infty] .
\ee
The only subtlety in this formula, explained in \autoref{sec:absolute-convergence-radius}, is that the number of terms with a given instanton number $|\vec{Y}|=k$ grows subexponentially with~$k$.

We will also need the notation $\floor{x}$ (resp.\ $\ceil{x}$) for the largest (resp.\ smallest) integer $\leq x$ (resp.\ $\geq x$), and notation for the fractional part and the distance to the nearest integer,
\be
\{x\} \coloneqq x - \floor{x} \in [0,1) , \qquad
\dist(x,\ZZ) \coloneqq \min(\{x\}, \{-x\}) .
\ee
It is also convenient to denote subintervals of the integers~$\ZZ$ or rationals~$\QQ$ by a subscripted inequality, for instance
\be
\ZZ_{\geq k} = \ZZ\cap [k,+\infty) , \qquad
\QQ_{<\alpha} = \QQ\cap (-\infty,\alpha) .
\ee

\subsection{Convergence/divergence strategy}
\label{sec:strategy}

Upper bounds of the form $\Rabs\leq R_1$ are obtained by showing that the absolute Nekrasov sum diverges for larger radii, by exhibiting a specific sequence of terms such that $Z_{\vec{Y}}$ grows faster than $R^{-|\vec{Y}|}$ for all $R>R_1$.
In \autoref{sec:Rabs_le_1} the easy bound
\be
\Rabs\leq 1
\ee
is shown for all cases of interest.
Stronger bounds for specific non-generic~$b^2$ appear in \autoref{sec:divergence}.

Lower bounds on~$\Rabs$ require placing upper bounds on~$Z_{\vec{Y}}$, hence lower bounds on the denominators in~\eqref{Zstar-21}.
The diagonal factors $Z_{II}$ involve $Z_{\lambda\lambda}(a,m)$ for $a=\eps_2$ and $a=-\eps_1$.
Depending on $\vec{Y}$, their denominators can then take arbitrary values of the form $-k\eps_1+l\eps_2$ for $k,l\geq 0$, with the exception of $(k,l)=(0,0)$ due to the shifts by $a\in\{\eps_2,-\eps_1\}$.
For generic~$b^2$, this lattice of values has no accumulation point, whereas for $b^2\in(0,+\infty)$ these values either include zero (for $b^2$ rational) or are dense in the line $\eps_2\RR$.
This leads us to distinguish several cases, treated separately.

The degeneracy assumption on Coulomb branch parameters ensures that none of the denominators in $Z_{IJ}$ for $I\neq J$ vanish.  Moreover, they are all bounded below uniformly by the distance $\dist(a_I-a_J,\eps_1\ZZ+\eps_2\ZZ)$, which is non-zero since $a_I-a_J$ is assumed to not belong to the closure of the lattice $\eps_1\ZZ+\eps_2\ZZ$.
This bound was used in~\cite{2212.06741}, together with a summation formula for certain bounds on~$Z_{II}$, to get the radius~\eqref{Arnaudo-Rabs}.
Our key improvement is to track better how many denominators can be small in the off-diagonal factors~$Z_{IJ}$.
The method extends to diagonal factors~$Z_{II}$ and to the regime $b^2>0$ with some modifications, avoiding the need for the summation formula in~\cite{2212.06741}.
In \autoref{sec:optimal-convergence} we prove that $\Rabs=1$ for $b^2\in\CC\setminus[0,+\infty)$.
\begin{itemize}
\item For $b^2\not\in\RR$, we use that $\bigl|-\eps_1 (\mu'_j-i) + \eps_2 (\lambda_i-j)\bigr|$ is larger than a multiple of the integer $|\mu'_j-i+\lambda_i-j|$ and that there are at most $O(|\lambda|^{1/2})$ boxes with any given value of this integer, which eventually allows us to bound $\log Z_{\lambda\mu}(a,m)$ above by $O\bigl(|\lambda|^{1/2}\log(|\lambda|+|\mu|)\bigr)$, sufficient to deduce $\Rabs\geq 1$.

\item The case $b^2<0$ is similar with only technical differences.  The comparison with $|\mu'_j-i+\lambda_i-j|$ no longer holds (because $\mu'_j-i$ has no sign), so we use the real $-\eps_1 (\mu'_j-i) + \eps_2 (\lambda_i-j)$ directly and show that there are at most $O(|\lambda|^{1/2})$ boxes for which it is within each interval of width $\min(1,-b^2)$.  This leads to the same bounds, hence $\Rabs\geq 1$.
\end{itemize}

Before the $b^2>0$ case can be discussed, \autoref{sec:sup-Brjuno} is devoted to notions of continued fractions, rational approximations of irrational numbers, and especially the exponential type $B_{\sup}(x)$ of an irrational~$x$, which is the infimum of numbers $B>0$ such that $e^{Bn}\dist(nb^2,\ZZ)$ has a positive infimum over $n\geq 1$.

For irrational $b^2>0$ of finite exponential type (the generic case), \autoref{sec:badly-approx} provides a positive lower bound on~$\Rabs$ in terms of~$B_{\sup}(b^2)$ by measuring how closely $-\eps_1 (\lambda'_j-i) + \eps_2 (\lambda_i-j)$ may approach zero as the partition size increases.

For irrational $b^2>0$ of finite or infinite exponential type, we exhibit in \autoref{sec:well-approx} specific terms $Z_{\vec{Y}}$ based on rational approximations of~$b^2$, which involve very small denominators hence which grow quickly with the instanton number~$|\vec{Y}|$.
This leads to an upper bound on the radius~$\Rabs$ in terms of the exponential type~$B_{\sup}(b^2)$.  In the case of infinite exponential type (super-exponentially well approximable irrationals), this upper bound vanishes, that is, $\Rabs=0$, which means that the absolute Nekrasov sum diverges for any $\qinst\neq 0$.

Finally, for rational $b^2>0$, in \autoref{sec:divergence-rational}, we find specific terms $Z_{\vec{Y}}$ that have poles at that value of~$b^2$, making the sum over partitions ill-defined.

\subsection{Divergence beyond the unit disk}
\label{sec:Rabs_le_1}

The reader is invited to jump to \autoref{sec:optimal-convergence} if they are interested in convergence results, namely lower bound on~$\Rabs$.
The techniques developped here are not reused until \autoref{sec:divergence}.
The present section establishes the expected bound $\Rabs\leq 1$, under the mild genericity assumptions used in Theorem~\ref{thm:2star}, excluding the case of positive rational $b^2>0$ (for which the sum is ill-defined).
For $1\leq I<J\leq N$ and $s=-1,0,1$ we thus assume
\beall{rad1-basic-genericity}
b^2 & \in \CC \setminus \QQ ,
& m & \not\in\overline{\eps_1\ZZ+\eps_2\ZZ} ,
& a_I - a_J & \not\in\overline{\eps_1\ZZ+\eps_2\ZZ} ,
\\
\text{or} \quad
b^2 & \in \QQ \cap (-\infty,0) ,
& m & \not\in\eps_1\ZZ+\eps_2\ZZ ,
& s m + a_I - a_J & \not\in\eps_1\ZZ+\eps_2\ZZ ,
\ee
where $\overline{\eps_1\ZZ+\eps_2\ZZ}$ is the closure of that lattice, equal to the lattice except for $b^2\in\RR\setminus\QQ$.

The strategy is to show that some terms~$Z_{\vec{Y}}$ do not decay exponentially with~$|\vec{Y}|$.
The precise choice of term is mostly irrelevant, except for one difficulty, which is that some numerators might vanish when $a_I-a_J+m$ belongs to $\eps_1\ZZ+\eps_2\ZZ$.
After giving a first series of terms that works generically (up to exchanging $\eps_1\leftrightarrow\eps_2$), we present a more elaborate choice of terms to cover the special case $m\in(\eps_1+\eps_2)\RR$ with $b^2\in\CC\setminus\RR$.

\paragraph{A choice that works half of the time.}

Let us solve first the problem under the extra genericity condition that, for some fixed index $I_0$ and for all~$J$,
\bel{rad1-extra-genericity}
a_J-a_{I_0}+m-\eps_1-\eps_2\not\in\eps_2\ZZ_{\geq 0}
\quad\text{and}\quad
a_J-a_{I_0}-m\not\in\eps_2\ZZ_{\geq 0} .
\ee
We return afterwards to the question of how this can be ensured under the condition~\eqref{rad1-basic-genericity} by an exchange of $\eps_1,\eps_2$, except for a special case for which a different configuration~$\vec{Y}$ is needed.

Consider a specific family of terms, $\vec{Y}=(\dots,\emptyset,\lambda,\emptyset,\dots)$ with $Y_{I_0}=\lambda=(k)$ being a one-term partition.  Then the term associated to $\vec{Y}$ in~\eqref{Zstar-21} is
\be
Z_{I_0I_0} = \prod_{j=1}^k \Bigl(\frac{j\eps_2-m}{j\eps_2} \, \frac{j\eps_2-\eps_1-\eps_2+m}{j\eps_2-\eps_1-\eps_2}\Bigr)
= \prod_{j=1}^k \Bigl(1 - \frac{m(m-\eps_1-\eps_2)}{j\eps_2(j\eps_2-\eps_1-\eps_2)}\Bigr)
\ee
multiplied by the following factor for each $J\neq I_0$,
\be
\begin{aligned}
Z_{I_0 J} & =
\prod_{j=1}^k \Bigl(\frac{a_{I_0}-a_J+\eps_1+j\eps_2-m}{a_{I_0}-a_J+\eps_1+j\eps_2} \,
\frac{a_{I_0}-a_J-\eps_2+j\eps_2+m}{a_{I_0}-a_J-\eps_2+j\eps_2}\Bigr)
\\
& = \prod_{j=1}^k \Bigl(1-\frac{m(m-\eps_1-\eps_2)}{(a_{I_0}-a_J+\eps_1+j\eps_2)(a_{I_0}-a_J-\eps_2+j\eps_2)}\Bigr) .
\end{aligned}
\ee
Under the genericity conditions \eqref{rad1-basic-genericity} and~\eqref{rad1-extra-genericity}, none of the factors vanish, so the logarithm is meaningful and we have
\beall{logZexam-212}
\log|Z_{\vec{Y}}|
& = \sum_{j=1}^k \log\Bigl|1-\frac{m(m-\eps_1-\eps_2)}{j\eps_2(j\eps_2-\eps_1-\eps_2)}\Bigr| \\
& \quad + \sum_{J\neq I_0}^N \sum_{j=1}^k \log\Bigl|1-\frac{m(m-\eps_1-\eps_2)}{(a_{I_0}-a_J+\eps_1+j\eps_2)(a_{I_0}-a_J-\eps_2+j\eps_2)}\Bigr| .
\ee
As $j\to +\infty$ the summand behaves as $O(1/j^2)$ so that the series converges.
(In fact, $Z_{\vec{Y}}$ has a limit.)
Thus, $(-1/k)\log|Z_{\vec{Y}}|\to 0$ as $k\to+\infty$, and by~\eqref{Rdef-Z},
\be
\Rabs \leq \exp(0) = 1 .
\ee

\paragraph{The choice of index~$I_0$.}
We wish to ensure~\eqref{rad1-extra-genericity}.
Consider a directed graph with vertices $I=1,\dots,N$ and with an edge $I\to J$ if
\bel{rad1-graph}
a_J-a_I \in \{m,\eps_1+\eps_2-m\} + \eps_2\ZZ_{\geq 0} .
\ee
If the graph has no loop, then by following successive edges we eventually reach a vertex~$I_0$ with no outgoing edge.  The lack of outgoing edge then precisely means that~\eqref{rad1-extra-genericity} holds.

If instead the graph has a loop, then by summing the $a_J-a_I$ for each edge $I\to J$ in the loop, we obtain
\bel{rad1-klp-integers}
0 = \sum_{I\to J\text{ in loop}} (a_J-a_I)
= k m + l(\eps_1+\eps_2-m) + p \eps_2
\ee
for integers $k,l,p\geq 0$, where $k$ counts the number of edges with $a_J-a_I \in m + \eps_2\ZZ_{\geq 0}$ and $l$ counts the number of edges with $a_J-a_I \in (\eps_1+\eps_2-m) + \eps_2\ZZ_{\geq 0}$, while $p$ accounts for the multiples of~$\eps_2$ in $\eps_2\ZZ_{\geq 0}$.  Since the loop has at least one edge, $(k,l)\neq(0,0)$.

The equation~\eqref{rad1-klp-integers} has solutions with $k=l\neq 0$ if and only if $\eps_1+\eps_2$ is a non-positive rational multiple of $\eps_2$, namely $b^2\in\QQ_{\leq -1}$.  The equation~\eqref{rad1-klp-integers} for $k\neq l\geq 0$ is equivalent to
\be
\sign(l-k)(m-(\eps_1+\eps_2)/2)
= \frac{k+l}{2|l-k|} (\eps_1+\eps_2) + \frac{p}{|l-k|} \eps_2 ,
\ee
and the existence of integer solution is equivalent to the existence of rational solutions (by rescaling $k,l,p$).  The positivity $k,l\geq 0$ implies $k+l\geq|l-k|$ so the first coefficient is at least $1/2$ while the other one is simply non-negative.
Altogether, \eqref{rad1-klp-integers}~has (integer) solutions if and only if
\bel{rad1-exist-integers}
b^2 \in \QQ_{\leq -1} ,
\quad \text{or} \quad
m-\frac{\eps_1+\eps_2}{2} \in \pm \bigl( \eps_2\QQ_{\geq 0} + (\eps_1+\eps_2) \QQ_{\geq 1/2} \bigr) ,
\ee
where the first case corresponds to $k=l\neq 0$, and the second case to $k\lessgtr l$.  The last set is a subset of the set of linear combinations of $\eps_2$ and $\eps_1+\eps_2$ with rational coefficients\footnote{The condition could be refined by noting in~\eqref{rad1-klp-integers} that $k+l\leq N$, so that the rational numbers in~\eqref{rad1-exist-integers} can be restricted to have denominators at most~$2N$.} that are either both non-negative or both non-positive.
If \eqref{rad1-exist-integers} does not hold, then \eqref{rad1-klp-integers} cannot have non-negative integer solutions with $(k,l)\neq (0,0)$, so the graph cannot have any loop.
As explained at the outset below~\eqref{rad1-graph}, the absence of loops ensures the extra genericity~\eqref{rad1-extra-genericity}, as wanted.

\paragraph{Case distinction.}

We now conclude in almost all cases.
\begin{itemize}
\item First case, $b^2\in\CC\setminus\RR$ and $m\not\in(\eps_1+\eps_2)\RR$.
  Then $\eps_1\RR$, $\eps_2\RR$ and $(\eps_1+\eps_2)\RR$ are three distinct lines through the origin, cutting the plane into six angular regions,
  \beal
  R_{12} & = \eps_1[0,+\infty)-\eps_2[0,+\infty) , \\
  R_{1*} & = \eps_1[0,+\infty)+(\eps_1+\eps_2)[0,+\infty) , \\
  R_{2*} & = \eps_2[0,+\infty)+(\eps_1+\eps_2)[0,+\infty) ,
  \ee
  and their opposites $R_{21}=-R_{12}=-\eps_1[0,+\infty)+\eps_2[0,+\infty)$ and $R_{*1}=-R_{1*}$ and $R_{*2}=-R_{2*}$.
  By assumption on~$m$, the point $m-(\eps_1+\eps_2)/2$ lies either outside $R_{1*}\cup R_{*1}$ or outside $R_{2*}\cup R_{*2}$ since their intersection is $(\eps_1+\eps_2)\RR$.
  If it lies outside $R_{2*} \cup R_{*2}$ then \eqref{rad1-exist-integers} does not hold.
  If it lies outside $R_{1*} \cup R_{*1}$ then one must swap $\eps_1\leftrightarrow\eps_2$ before applying the results above, which amounts to considering terms $\vec{Y}=(\dots,\emptyset,\lambda,\emptyset,\dots)$ with $Y_{I_0}=\lambda=(1^k)$ instead of $(k)$.

\item Second case, $b^2\in\CC\setminus\RR$ and $m\in(\eps_1+\eps_2)\RR$, treated momentarily.

\item Third case, $b^2\in\RR\setminus\QQ$, so the lattice $\eps_1\ZZ+\eps_2\ZZ$ is dense in the line $\eps_1\RR=\eps_2\RR=(\eps_1+\eps_2)\RR$.  By the basic genericity assumption~\eqref{rad1-basic-genericity} that $m$ does not lie in that line, \eqref{rad1-exist-integers}~does not hold.

\item Fourth case, $b^2\in\QQ_{<0}$.  Then the genericity condition~\eqref{rad1-basic-genericity} states that $a_I-a_J\pm m\not\in\eps_1\ZZ+\eps_2\ZZ$ so the ``extra'' genericity~\eqref{rad1-extra-genericity} holds, namely the numerators for the first configuration $\vec{Y}$ we considered do not vanish.

\item Fifth case, $b^2\in\QQ_{>0}$.  This was excluded by our assumption~\eqref{rad1-basic-genericity}.  See \autoref{sec:divergence-rational} for specific terms that are singular.
\end{itemize}
In all cases except the second one, we have thus found a sequence of terms $Z_{\vec{Y}}$ that do not decay as $|\vec{Y}|\to+\infty$, so $\Rabs\leq 1$.

\paragraph{A different configuration.}

We are now left with the non-generic case $b^2\in\CC\setminus\RR$ and $m\in(\eps_1+\eps_2)\RR$, for which we must select new terms.
Decompose each Coulomb branch parameter as
\bel{aIeps2eps12}
a_I = \alpha_I \eps_2 + \beta_I (\eps_1+\eps_2) , \qquad \alpha_I,\beta_I\in\RR ,
\ee
and let $S$ be the set of indices $I$ whose $\alpha_I$ is maximal,
\be
S = \bigl\{ I \bigm| \alpha_I = \max_J \alpha_J \bigr\} \subseteq \{1,\dots,N\} .
\ee
Consider now $\vec{Y}$ with $Y_I = \lambda = (k)$ for $I\in S$ and $Y_I = \emptyset$ otherwise.  Then~\eqref{Zstar-21} evaluates to
\beall{Zstar-21-0l0l}
& Z_{\vec{Y}}
= \prod_{I,J\in S} \prod_{j=1}^k
\Bigl(\frac{a_I-a_J + j \eps_2 - m}{a_I-a_J + j \eps_2}\Bigr)
\Bigl(\frac{a_I-a_J-(\eps_1+\eps_2) + j \eps_2+m}{a_I-a_J-(\eps_1+\eps_2) + j \eps_2}\Bigr)
\\
& \quad \times \prod_{I\in S} \prod_{J\not\in S} \prod_{j=0}^{k-1}
\Bigl(\frac{a_I-a_J+(\eps_1+\eps_2) + j \eps_2-m}{a_I-a_J+(\eps_1+\eps_2) + j \eps_2}\Bigr)
\Bigl(\frac{a_I-a_J + j \eps_2+m}{a_I-a_J + j \eps_2}\Bigr) .
\ee
The numerators in the first line involve the difference $a_I-a_J$, which lies in the line $(\eps_1+\eps_2)\RR$ since $\alpha_I=\alpha_J$ for $I,J\in S$, so these numerators all lie in $(\eps_1+\eps_2)\RR+\{1,\dots,k\}\eps_2$ and cannot vanish.
The numerators in the second line also have a positive multiple of $\eps_2$ when decomposed in the form~\eqref{aIeps2eps12} because $\alpha_I>\alpha_J$ for $I\in S$ and $J\not\in S$.
The same logic as~\eqref{logZexam-212} then works: for each $I,J,j$ the factor in~\eqref{Zstar-21-0l0l} is non-zero and has the form $1+O(1/j^2)$ as $j\to +\infty$, so that the product over $j\in\{1,\dots,k\}$ has a finite limit as $k\to+\infty$.  This means that $Z_{\vec{Y}}$ has a finite limit as $k\to+\infty$ and $\Rabs\leq 1$ in this case too.

\section{Optimal convergence}
\label{sec:optimal-convergence}

It is now time to derive the flagship result: absolute convergence of the Nekrasov instanton partition function (one-point torus conformal blocks) with
\be
\Rabs=1 \qquad \text{for } b^2\in\CC\setminus[0,+\infty) .
\ee
Since $\Rabs\leq 1$ by the previous section, we need to show now that all terms~$Z_{\vec{Y}}$ grow at most sub-exponentially with~$|\vec{Y}|$.
We distinguish $b^2\in\CC\setminus\RR$ from $b^2<0$.

In both cases, a consequence of $|1+X|\leq 1+|X|\leq\exp|X|$ is useful:
\beall{logZ22}
\log\bigl|Z_{\lambda\mu}(a,m)\bigr|
& \leq \sum_{(i,j)\in\lambda} \Bigl| \frac{m}{a-\eps_1 (\mu'_j-i) + \eps_2 (\lambda_i-j)} \Bigr| ,
\ee
with the understanding that $\log 0 = -\infty$ is less than any real number.

\subsection{Non-real \texorpdfstring{$b^2$}{b\texttwosuperior}}

\paragraph{Lower bounds on denominators.}

There are two cases of interest to~\eqref{Zstar-21}.
On the one hand, the diagonal factors $Z_{II}$ involve $Z_{\lambda\lambda}(a,m)$ for $a=\eps_2$ or $a=-\eps_1$, whose denominators belong to $(\eps_1\ZZ+\eps_2\ZZ)\setminus\{0\}$ as explained in \autoref{sec:strategy}.
For non-real $b^2$, $\eps_1\ZZ+\eps_2\ZZ$ is a discrete lattice whose points have some minimum pairwise distance $d_{\eps_1,\eps_2}$, which is thus a lower bound on (the absolute value of) denominators in~\eqref{logZ22}.
On the other hand, the off-diagonal factors have $a=a_I-a_J\not\in\eps_1\ZZ+\eps_2\ZZ$ by assumption and, again because this lattice is discrete, the distance of $a$ to the lattice is non-zero and serves as a lower bound on all denominators.
Overall, in both cases of interest to us, we have a lower bound on denominators,
\bega
|a-\eps_1 (\mu'_j-i) + \eps_2 (\lambda_i-j)| \geq D , \\
D \coloneqq \min\Bigl( \dist\bigl(0,(\eps_1\ZZ+\eps_2\ZZ)\setminus\{0\}\bigr) , \min_{I<J} \dist\bigl(a_I-a_J,\eps_1\ZZ+\eps_2\ZZ\bigr) \Bigr) > 0 .
\ee

Next we derive a better lower bound that grows with $\sigma_{ij}=\mu'_j-i+\lambda_i-j$:
\beall{bound-by-sigmaij}
\bigl|-\eps_1 (\mu'_j-i) + \eps_2 (\lambda_i-j)\bigr|
& = |\eps_1+\eps_2| \Bigl|\frac{\eps_2}{\eps_1+\eps_2} \sigma_{ij} + i - \mu'_j\Bigr|
\\
& \geq |\eps_1+\eps_2| \Bigl|\Im\Bigl(\frac{\eps_2}{\eps_1+\eps_2} \sigma_{ij}\Bigr)\Bigr| = D_{\eps_1,\eps_2} |\sigma_{ij}| ,
\ee
where we used that $i-\mu'_j$ and $\sigma_{ij}$ are real, and where
\be
D_{\eps_1,\eps_2}
\coloneqq |\eps_1+\eps_2| \Bigl|\Im\Bigl(\frac{\eps_2}{\eps_1+\eps_2} \Bigr) \Bigr|
= \Bigl| \frac{\eps_2\Im(b^2)}{1+b^2}\Bigr| ,
\ee
which is positive for non-real~$b^2$.
By the triangle inequality we deduce the bound
\bel{lower26}
|a-\eps_1 (\mu'_j-i) + \eps_2 (\lambda_i-j)| \geq \max\bigl(D , D_{\eps_1,\eps_2} |\sigma_{ij}| - |a| \bigr)
\ee
which improves on the previous one for large enough $|\sigma_{ij}|$.

\paragraph{A sub-linear bound on the logarithm.}

The number of boxes with a given value $\sigma_{ij}=s$ cannot grow too large.
Note that $\sigma_{ij}$ is monotonically strictly decreasing in~$i$ and~$j$ (because $\mu'_j$ and $\lambda_i$ are non-increasing), so two boxes $(i,j),(s,t)\in\lambda$ with $\sigma_{ij}=\sigma_{st}$ cannot be in the same row nor in the same column and must have $i<s$ and $j>t$ or the opposite.
Assume now that $q$ boxes $(i_k,j_k)\in\lambda$ for $k=1,\dots,q$ (sorted by non-decreasing row number~$i_k$) share the same value $\sigma_{i_kj_k}=s$.
Then one must have $i_1<\dots<i_q$ and $j_1>\dots>j_q$, hence $i_k\geq k$ and $j_k\geq q+1-k$, namely $\lambda$ contains a staircase partition of side~$q$, so that $|\lambda| \geq q(q+1)/2\geq q^2/2$, which implies a bound on the cardinal
\bel{num-sigma}
\Card\{(i,j)\in\lambda\mid\sigma_{ij}=s\} \leq \sqrt{2|\lambda|}
\ee
for any $s\in\ZZ$.
Organizing the boxes in~\eqref{logZ22} by the value of~$\sigma_{ij}$, and using the lower bound~\eqref{lower26} on each denominator, we deduce
\bel{logZ-28}
\log\bigl|Z_{\lambda\mu}(a,m)\bigr|
\leq \sum_{s=-|\lambda|}^{|\lambda|+|\mu|} \frac{|m|\Card\{(i,j)\in\lambda\mid\sigma_{ij}=s\}}{\max(D,D_{\eps_1,\eps_2}|s|-|a|)} .
\ee
The (non-optimized) range of~$s$ is found by noting that $\sigma_{ij}\leq\mu'_j+\lambda_i\leq|\lambda|+|\mu|$ and $\sigma_{ij}\geq -i\geq -|\lambda|$.
The numerator is bounded above by $|m|\sqrt{2|\lambda|}$ for any~$s$.
On the other hand, the sum of $1/\max(D,D_{\eps_1,\eps_2}|s|-|a|)$ can be split into a finite range $|s|\leq K$ with
\be
K \coloneqq \max\Bigl(\Bigfloor{\frac{D+|a|}{D_{\eps_1,\eps_2}}} , \Bigceil{\frac{|a|}{D_{\eps_1,\eps_2}}} \Bigr) ,
\ee
and the rest, where the sum is a harmonic series that grows logarithmically.
We deduce (where $H_n=1+1/2+\dots+1/n$ are harmonic numbers)
\beal
\log\bigl|Z_{\lambda\mu}(a,m)\bigr|
& \leq |m|\sqrt{2|\lambda|} \biggl( \frac{2K+1}{D} + \frac{H_{|\lambda|}+H_{|\lambda|+|\mu|}}{D_{\eps_1,\eps_2}} \biggr)
\\
& \leq C |\lambda|^{1/2} \log\bigl(2+|\lambda|+|\mu|\bigr)
\ee
for a sufficiently large constant~$C$ (depending on $\eps_1,\eps_2,m,a$).
(The artificial shift by~$2$ ensures that the logarithm is positive even for empty diagrams.)

\paragraph{Absolute convergence radius.}

Going back to the expression~\eqref{Rdef-Z} of the convergence radius, we have (recall that $Z_{IJ}$ is a product of two $Z_{\lambda\mu}$ factors)
\beal
- \log \Rabs
& = \limsup_{\vec{Y}} \frac{1}{|\vec{Y}|} \sum_{I,J=1}^N \log |Z_{IJ}|
\\
& \leq 2 N^2 C \limsup_{\vec{Y}} |\vec{Y}|^{-1/2} \log\bigl(2 + 2 |\vec{Y}|\bigr) = 0 .
\ee
Combining with $\Rabs\leq 1$ from \autoref{sec:Rabs_le_1}, we get $\Rabs=1$ as desired.

\subsection{Negative \texorpdfstring{$b^2$}{b\texttwosuperior}}

\paragraph{Lower bounds on denominators.}

We now turn to the case $b^2<0$.  Our starting point is again the easy bound~\eqref{logZ22} on $Z_{\lambda\mu}(a,m)$.
Denominators for the diagonal~$Z_{II}$, divided by~$\eps_2$, take the form $-kb^2+l$ for $k,l\geq 0$ and $(k,l)\neq(0,0)$, which are bounded below by $\min(-b^2,1)>0$, hence the absolute value of these denominators is at least $\min(|\eps_1|,|\eps_2|)$.
Denominators for the off-diagonal $Z_{IJ}$ are bounded below by the distance between $a_I-a_J$ and the lattice $\eps_1\ZZ+\eps_2\ZZ$, which is positive since we have assumed that $a_I-a_J$ is away from the lattice's closure.
We deduce that for both cases of interest to us,
\bega
|a-\eps_1 (\mu'_j-i) + \eps_2 (\lambda_i-j)| \geq D , \\
D \coloneqq \min\Bigl( |\eps_1|, |\eps_2|,  \min_{I<J} \dist\bigl(a_I-a_J,\eps_1\ZZ+\eps_2\ZZ\bigr) \Bigr) > 0 .
\ee
As in the previous case, we seek an improved bound that depends on the box $(i,j)$.
The bound~\eqref{bound-by-sigmaij} by~$\sigma_{ij}$ does not hold because $\mu'_j-i$ has no sign.
Rather than $\sigma_{ij}$ we introduce the more elaborate integer
\be
\tau_{ij} = \Bigfloor{\frac{(-b^2)(\mu'_j-i)+(\lambda_i-j)}{\min(1,-b^2)}}.
\ee
To use this integer $\tau_{ij}$, we write a lower bound
\be
|\eps_1 (\mu'_j-i) + \eps_2 (\lambda_i-j)|
\geq |\eps_2| \min(1,-b^2) \biggl| \frac{(-b^2)(\mu'_j-i)+(\lambda_i-j)}{\min(1,-b^2)} \biggr|
\ee
and we note that the last absolute value is bounded below by $|\tau_{ij}| - \onebf_{\tau_{ij}<0}$ since the ratio is in the interval $[\tau_{ij},\tau_{ij}+1)$.
Altogether,
\be
|a-\eps_1 (\mu'_j-i) + \eps_2 (\lambda_i-j)| \geq \max\Bigl(D ,
\min(|\eps_1|,|\eps_2|) \bigl(|\tau_{ij}| - \onebf_{\tau_{ij}<0}\bigr) - |a| \Bigr) .
\ee

\paragraph{A speedy proof.}

The proof then goes through as in the previous case, with the combination $\tau_{ij}$ replacing~$\sigma_{ij}$.
The key point is that $\tau_{ij}$ is strictly monotonically decreasing in $i,j$ because the argument of $\floor{\ }$ decreases by at least $1/\min(1,-b^2)\geq 1$ when $j$ is increased, and $-b^2/\min(1,-b^2)\geq 1$ when $i$ in increased.
Thus, boxes with the same $\tau_{ij}$ must be in different rows and columns with the appropriate opposite order, which again shows that the number of boxes with a given $\tau_{ij}=s$ is at most $\sqrt{2|\lambda|}$.
The upper bound~\eqref{logZ-28} becomes
\be
\log\bigl|Z_{\lambda\mu}(a,m)\bigr|
\leq \sum_{s=\frac{-|\lambda|(-b^2)}{\min(1,-b^2)}}^{\frac{(|\lambda|-b^2|\mu|)}{\min(1,-b^2)}} \frac{|m|\Card\{(i,j)\in\lambda,\tau_{ij}=s\}}{\max\Bigl(D,\min(|\eps_1|,|\eps_2|)\bigl(|s|-\onebf_{s<0}\bigr)-|a|\Bigr)} .
\ee
Apart from a finite middle part, the sum is once more a harmonic sum with bounds that grow linearly in $|\lambda|+|\mu|$, so
\be
\log\bigl|Z_{\lambda\mu}(a,m)\bigr|
\leq C |\lambda|^{1/2} \log\bigl(2+|\lambda|+|\mu|\bigr)
\ee
for a sufficiently large constant~$C$ (depending on $\eps_1,\eps_2,m,a$).
Since this bound is sublinear, it ensures that $-|\vec{Y}|^{-1}\log|Z_{\vec{Y}}|$ tends to zero so that $\log\Rabs\geq 0$ and we conclude that $\Rabs=1$ exactly.

\section{A variant of Brjuno numbers}
\label{sec:sup-Brjuno}

To treat the case of non-negative~$b^2$, one must distinguish rational values of~$b^2$ and two types of irrational values, distinguished by how closely they are approached by rationals of a given denominator.
The required relaxed variant of Brjuno numbers is introduced in this section.  This relies on a preliminary discussion of continued fractions.
To keep the story general, the positive irrational number $b^2$ is denoted~$x$ instead.

\subsection{Continued fractions}

\paragraph{Basics.}

We shall state with only minimal proof some well-known properties of continued fractions.
Irrational numbers $x\in\RR\setminus\QQ$ can be represented bijectively by a continued fraction $x=[a_0;a_1,\dots]$ with integer coefficients $a_0\in\ZZ$ and $a_n\geq 1$ for $n\geq 1$.  These coefficients are defined iteratively by (for $n\geq 0$)
\begin{equation}
  x_0 = x , \qquad a_n = \lfloor x_n\rfloor , \qquad x_{n+1} = 1/\{x_n\} ,
\end{equation}
where $\{x_n\}$ is the fractional part.
By construction one has (for $n\geq 0$)
\begin{equation}\label{x-continued-xn}
  x = a_0 + \frac{1}{a_1+\frac{1}{a_2+\frac{1}{\cdots+\frac{1}{a_n+\{x_n\}}}}} .
\end{equation}

It is also useful to introduce \emph{convergents} $p_n/q_n$ by $(p_{-2},q_{-2})=(0,1)$ and $(p_{-1},q_{-1})=(1,0)$ and the three-term recursion relation
\begin{equation}
  p_n = a_n p_{n-1} + p_{n-2} , \qquad
  q_n = a_n q_{n-1} + q_{n-2} .
\end{equation}
This leads for instance to $(p_0,q_0)=(a_0,1)$ and $(p_1,q_1)=(a_1 a_0+1,a_1)$.

An easy induction shows that
\begin{equation}\label{cont-frac-recursion}
  x = \frac{p_nx_{n+1} + p_{n-1}}{q_nx_{n+1} + q_{n-1}} .
\end{equation}
By taking the limit $\{x_n\}=1/x_{n+1}\to 0$ in \eqref{x-continued-xn} and~\eqref{cont-frac-recursion}, one learns that $p_n/q_n$ is simply the rational number obtained by truncating the continued fraction,
\begin{equation}
  \frac{p_n}{q_n} = [a_0;a_1,\dots,a_n] = a_0 + \frac{1}{a_1+\frac{1}{a_2+\frac{1}{\cdots+\frac{1}{a_n}}}} .
\end{equation}

Another consequence of~\eqref{cont-frac-recursion} is that $(q_{n-1}x-p_{n-1}) / (q_nx-p_n) = -x_{n+1} < 0$, so that $(-1)^n(q_nx-p_n)$ has a constant sign, positive.  The even (resp.\ odd) convergents are monotonically increasing (resp.\ decreasing) towards~$x$:
\begin{equation}
  \frac{p_0}{q_0} < \frac{p_2}{q_2} < \frac{p_4}{q_4} < \dots < x < \dots < \frac{p_5}{q_5} < \frac{p_3}{q_3} < \frac{p_1}{q_1} .
\end{equation}

\paragraph{Equidistribution.}

The next lemma states that if $x\simeq p/q$ then many groups of $q$~consecutive fractional parts $\{kx\},\allowbreak\{(k+1)q\},\allowbreak\dots,\allowbreak\{(k+q-1)x\}$ are equidistributed in the sense that there is one in each interval $[0,1/q),\allowbreak[1/q,2/q),\allowbreak\dots,\allowbreak[(q-1)/q,1)$.
In particular this applies to the convergents, with $(p,q,r,\epsilon)=(p_n,q_n,q_{n+1},(-1)^n)$, thanks to the bound~\eqref{estimate-quality-qnpn} below.
For an integer interval $I=\intset{k}{l}$ we denote $-I=\intset{-l}{-k}$.

\begin{lemma}[Equidistribution of irrational orbits]\label{lem:equidistribution}
  1. For any $x\in\RR$ and integer $q\geq 1$, for any $j\in\ZZ$, the residue of $\lfloor qjx\rfloor$ modulo~$q$ is
  \begin{equation}
    \varphi_{q,x}(j) \coloneqq \lfloor q\{jx\}\rfloor\in\intset{0}{q-1} .
  \end{equation}

  2. For any coprime integers $p\in\ZZ$ and $q\geq 1$, and any integer $r\geq 1$ and sign $\epsilon=\pm$ such that $0\leq \epsilon(qx-p)<1/r$, and for any integer interval $I$ with $q$~elements such that $\epsilon I$ lies entirely in the interval $\intset{-r}{-1}$ or $\intset{0}{r}$, the restriction $\varphi_{q,x}:I\to\intset{0}{q-1}$ is a bijection, given explicitly in \eqref{phiqx1} and~\eqref{phiqx2}.
\end{lemma}

\begin{proof}
  For the first point, $q\{jx\}=q(jx-\lfloor jx\rfloor)$ so $q\{jx\}$ and $qjx$ agree as reals modulo~$q$.
  Taking the floor on both sides gives $\varphi_{q,x}(j)=\lfloor qjx\rfloor \bmod{q}$ and we conclude by $0\leq q\{jx\}<q$.
  The second point is a variation on a proof in \url{https://mathoverflow.net/a/313290}, with a wider set of intervals requiring a case distinction.

  For $j\in\ZZ$, the residue of $qj(p/q)=pj$ modulo~$q$ is the integer
  \begin{equation}
    \varphi_{q,p/q}(j)\coloneqq q\{pj/q\}\in\intset{0}{q-1} .
  \end{equation}
  Since $p$ and $q$ are coprime, multiplication by~$p$ modulo~$q$ is a bijection of $\ZZ/q\ZZ$, so the restriction of $\varphi_{q,p/q}$ to any interval of $q$ consecutive integers is a bijection to $\intset{0}{q-1}$.  Thus, for $x=p/q$ we are done.  We henceforth assume a strict inequality $0<\epsilon(qx-p)<1/r$.

  Consider first the range $\epsilon I\subset\intset{0}{r}$.
  For any $j\in\epsilon\intset{1}{r}$, the condition on~$x$, multipled by the positive number $\epsilon j$, gives $0<j(qx-p)<\epsilon j/r\leq 1$, thus $jqx\in(pj,pj+1)$ so $\lfloor qjx\rfloor=pj$, and by taking residues modulo~$q$,
  \begin{equation}\label{phiqx1}
    \varphi_{q,x}(j) = \varphi_{q,p/q}(j) , \qquad j\in \epsilon\intset{0}{r} ,
  \end{equation}
  where we included the trivial case $j=0$ in the conclusion.
  The restriction of $\varphi_{q,x}$ to any $q$-element interval $I\subset\epsilon\intset{0}{r}$ is then simply the restriction of $\varphi_{q,p/q}$, which is bijective.

  For the range $\epsilon I\subset\intset{-r}{-1}$, we multiply the condition on~$x$ by a negative integer $\epsilon j\in\intset{-r}{-1}$ to get $-1\leq\epsilon j/r<j(qx-p)<0$, so that $jqx\in(pj-1,pj)$, which leads to
  \begin{equation}\label{phiqx2}
    \varphi_{q,x}(j) = \varphi_{q,p/q}(j) - 1 \mod q , \qquad j\in \epsilon\intset{-r}{-1} .
  \end{equation}
  The restriction of $\varphi_{q,x}$ to any $q$-element subinterval $I\subset\epsilon\intset{-r}{-1}$ is thus the composition of two bijections, $\varphi_{q,p/q}$~restricted to~$I$, and a shift by~$1$ modulo~$q$.
\end{proof}

\paragraph{Quality of approximation.}

From \eqref{cont-frac-recursion} and $q_np_{n-1} - p_nq_{n-1}=(-1)^n$ (proven by an easy induction) and $q_n x_{n+1} + q_{n-1} = q_{n+1} + q_n \{x_{n+1}\}$, one has
\begin{equation}\label{estimate-quality-qnpn}
  0 \leq (-1)^n (q_n x - p_n) = \frac{1}{q_{n+1}+q_n\{x_{n+1}\}} \in \Bigl( \frac{1}{2q_{n+1}}, \frac{1}{q_{n+1}} \Bigr] .
\end{equation}
This estimate of the quality of the approximation $x\simeq p_n/q_n$ will prove essential.
Rather than numerators and denominators of approximations to~$x$, it is often more convenient to consider the fractional parts $\{qx\}$ and $1-\{qx\}=\{-qx\}$ (for irrational~$x$), and collectively the distance $\dist(qx,\ZZ)=\min(\{qx\},\{-qx\})$.

\begin{lemma}\label{lem:dist-greater}
  For $j\in(q_n,q_{n+1})$ one has $\dist(jx,\ZZ)>\dist(q_nx,\ZZ)$.
\end{lemma}

\begin{proof}
  By~\eqref{estimate-quality-qnpn}, $\dist(q_nx,\ZZ)\leq 1/q_{n+1}$.
  It is thus sufficient to prove that $\{jx\}$ and $\{-jx\}$ are larger than this.
  We apply Lemma~\ref{lem:equidistribution}, for $(p,q,r,\epsilon)=(p_n,q_n,q_{n+1},(-1)^n)$, to find the integer parts of $q_n\{jx\}$ and $q_n\{-jx\}$, which depend only on the residue of $j$ modulo~$q_n$ and its sign.  Whenever these integer parts are non-zero, we are done.  There remains two tasks:
  \begin{itemize}
  \item to show $\{(-1)^n jx\}>1/q_{n+1}$ for $j=kq_n$ with $k\in\intset{2}{a_{n+1}}$;
  \item to show $\{(-1)^{n+1} jx\}>1/q_{n+1}$ for $j=kq_n+q_{n-1}$ with $k\in\intset{1}{a_{n+1}-1}$.
  \end{itemize}
  (The upper bounds on~$k$ come from $\lfloor q_{n+1}/q_n\rfloor=a_{n+1}$.)
  For the first task, observe that $\{(-1)^n jx\}=k\{(-1)^n q_n x\}$ since they agree modulo~$1$ and the latter is bounded above by $k/q_{n+1}<1$.  The desired lower bound then arises from the lower bound in~\eqref{estimate-quality-qnpn} and $k\geq 2$.
  For the second task, write $j=q_{n+1}-(a_{n+1}-k)q_n$ and observe that
  \begin{equation}
    \{(-1)^{n+1} jx\} = \{(-1)^{n+1}q_{n+1}x\} + (a_{n+1}-k) \{(-1)^n q_n x\}
  \end{equation}
  since they agree modulo~$1$ and the latter is bounded by $(1+a_{n+1}-k)/q_{n+1}<1$.
  The desired lower bound then comes from $a_{n+1}-k\geq 1$ and positivity of the first term.
\end{proof}

\subsection{Exponential type}

\paragraph{The notion of interest.}

A \emph{Brjuno number} is an irrational number~$x$ whose convergents $p_n/q_n$ are such that the Brjuno function takes a finite value,
\begin{equation}
  B(x) \coloneqq \sum_{n\geq 0} \frac{\log q_{n+1}}{q_n} < +\infty .
\end{equation}
This is a very large class, which only excludes irrational numbers that are extremely well approximated by rationals.
Replacing the sum by a supremum yields the notion of numbers of \emph{finite exponential type} (see, e.g.~\cite{2205.00253}), with several characterizations whose equivalence is proven below.\footnote{I have been unable to confirm how standard this terminology is, and propose the alternative \emph{sup-Brjuno numbers} to emphasize the link to Brjuno numbers.  As a side-note, restricting the first lim sup in Definition~\ref{def:expo-type} to only even or only odd~$n$, and correspondingly replacing $\dist(qx,\ZZ)$ by $\{qx\}$ or $\{-qx\}$ in other expressions appears to define an interesting notion too.}

\begin{definition}\label{def:expo-type}
  The \emph{exponential type} of an irrational number $x\in\RR\setminus\QQ$ is
  \[
    \begin{aligned}
      B_{\sup}(x)
      \, & \! \coloneqq \limsup_{n\to+\infty} \frac{\log q_{n+1}}{q_n}
      = \limsup_{q\to+\infty} \frac{\log(1/\dist(qx,\ZZ))}{q} \\
      & = \inf \Bigl\{B\in[0,+\infty)\Bigm| e^{Bq}\dist(qx,\ZZ) \xrightarrow{q\to+\infty} +\infty\Bigr\} \\
      & = \inf \Bigl\{B\in[0,+\infty)\Bigm| \dist(qx,\ZZ) \gtrsim e^{-Bq} ,  \ q\to+\infty \Bigr\} .
    \end{aligned}
  \]
  The number is said to be of \emph{finite exponential type} if that quantity is finite.  It is otherwise said to be of infinite exponential type, or \emph{super-exponentially well approximable} (by rationals).
\end{definition}

\begin{remark}
  Almost every irrational number~$x$ has $B_{\sup}(x)=0$, in the sense that the set $\{x\mid B_{\sup}(x)>0\}$ has measure zero.
  In fact, every Brjuno number has $B_{\sup}(x)=0$ since the convergence of the series defining $B(x)$ requires its terms to tend to zero.
  A number with $B_{\sup}(x)>0$ is necessarily a Liouville number (i.e., has rational approximations with $|x-p/q|<1/q^K$ for arbitrarily large~$K$).
  Examples are somewhat artificial, such as $x=\sum_{n\geq 0} 1/(2\uparrow\uparrow n)$ where $2\uparrow\uparrow 0 = 1$ and $2\uparrow\uparrow n=2^{2\uparrow\uparrow(n-1)}$ is the Knuth up arrow notation, which has $B_{\sup}(x)=\log 2$.
\end{remark}

\paragraph{Equivalence proof.}

The equality of these three expressions (which we denote temporarily by $B_1,B_2,B_3,B_4$ to express this proof) is straightforward.
\begin{itemize}
\item Proof of $B_1=B_2$.  For $q\in[q_n,q_{n+1})$, one has $\log(1/\dist(qx,\ZZ))/q \leq \log(1/\dist(q_nx,\ZZ)) / q_n$ by Lemma~\ref{lem:dist-greater} so the lim sup defining~$B_2$ is the same as the lim sup restricted to $\{q_n\mid n\geq 0\}$.
Observe then that $\dist(q_nx,\ZZ)\in(1/(2q_{n+1}),1/q_{n+1}]$ so
\be
\frac{1}{q_n} \log\frac{1}{\dist(q_nx,\ZZ)}
= \frac{1}{q_n} \log q_{n+1} + o(1) , \qquad n\to+\infty ,
\ee
which has $B_1$ as its lim sup.
\item Proof of $B_3\leq B_2$.  Assume by contradiction that $B_2<B_3$ and take $B_2<B<B'<B_3$.  Then for $q$~large enough, $\log(1/\dist(qx,\ZZ))/q<B$, namely $e^{Bq}\dist(qx,\ZZ)>1$, thus $e^{B'q}\dist(qx,\ZZ)\to+\infty$ as $q\to+\infty$, which contradicts the definition of~$B_3$.
\item Proof of $B_4\leq B_3$.  This is trivial since $e^{Bq}\dist(qx,\ZZ)\to+\infty$ implies $\dist(qx,\ZZ)\gtrsim e^{-Bq}$.
\item Proof of $B_2\leq B_4$.  Assume by contradiction that $B_4<B_2$.  Then by definition of~$B_4$ there is some $B\in(B_4,B_2)$ with $e^{Bq}\dist(qx,\ZZ)\geq 1/C$ for some positive constant $C>0$, so $(1/q)\log(1/\dist(qx,\ZZ))<B+(\log C)/q$, which tends to~$B$, thus contradicting the definition of~$B_2$.
\end{itemize}

\paragraph{Covariance properties.}
It is interesting to consider how $B_{\sup}$ behaves under simple transformations.

\begin{proposition}\label{prop:B-inverse}
  Given an irrational number $x\in\RR\setminus\QQ$ and an integer $l\in\ZZ$, one has
  \[
  B_{\sup}(x) = B_{\sup}(x+l) = B_{\sup}(-x) = |x| B_{\sup}(1/x) .
  \]
  As a result, for any matrix $\bigl(\begin{smallmatrix}a&b\\c&d\end{smallmatrix}\bigr)\in GL(2,\ZZ)$ one has
  \[
  B_{\sup}\Bigl(\frac{ax+b}{cx+d}\Bigr) = \frac{1}{|cx+d|} B_{\sup}(x) .
  \]
\end{proposition}

\begin{proof}
  The identity for a general $GL(2,\ZZ)$ matrix is consistent with composition of M\"obius transformations, namely products of matrices: if $z=(ay+b)/(cy+d)$ and $y=(a'x+b')(c'x+d')$ then
  \beal
  B_{\sup}(z)
  & = \frac{1}{|cy+d|} B_{\sup}(y) = \frac{1}{|cy+d|\,|c'x+d'|} B_{\sup}(x)
  \\
  & = \frac{1}{|(ca'+dc')x+(cb'+dd')|} B_{\sup}(x) ,
  \ee
  which is the expected relation since the denominator here is the denominator of the M\"obius transformation from $x$ to $z$ as expected.
  It is thus enough to check the transformation for a set of generators $x\mapsto -x$, $x\mapsto x+l$ for $l\in\ZZ$, and $x\mapsto 1/x$.
  The invariance of $B_{\sup}(x)$ under the first two generators is immediate since $\dist(qx,\ZZ)$ has these invariances for $q\in\ZZ$.

  To compute $B_{\sup}(1/x)$, assume $x>0$ without loss of generality, and consider successive convergents $p_n/q_n$ and $p_{n+1}/q_{n+1}$ for $n\geq 0$.
  They surround~$x$ in the sense that $(-1)^np_n/q_n < (-1)^n x < (-1)^n p_{n+1}/q_{n+1}$.
  One has
  \be
  0 < (-1)^n (q_n - p_n/x) < (-1)^n (q_np_{n+1} - p_nq_{n+1})/p_{n+1} = 1/p_{n+1} .
  \ee
  Now we take the second expression in Definition~\ref{def:expo-type} and restrict $q$ to $\{p_n\mid n\geq 0\}$ (thus getting a lower bound of the exponential type) to get
  \be
    \begin{aligned}
      B_{\sup}(1/x)
      & \geq \limsup_{n\to+\infty} \frac{1}{p_n} \log(1/\dist(p_n/x,\ZZ))
        \geq \limsup_{n\to+\infty} \frac{\log p_{n+1}}{p_n}
      \\
      & = \limsup_{n\to+\infty} \frac{\log(q_{n+1} x+o(1))}{q_n x+o(1)}
        = \frac{1}{x} \limsup_{n\to+\infty} \frac{\log q_{n+1}}{q_n}
        = \frac{1}{x} B_{\sup}(x) .
    \end{aligned}
  \ee
  To conclude, we observe that the same inequality holds with $x\to 1/x$, namely $B_{\sup}(x) \geq x B_{\sup}(1/x)$.
\end{proof}

\subsection{An average of logarithms}

The set of fractional parts of $\pm lx$ for integer $l\geq 1$ is known to satisfy some equidistribution properties stated in Lemma~\ref{lem:equidistribution}, with outliers that are closer to zero depending on how well~$x$ is approximated by rationals.  The equidistribution and the outliers can be measured by summing a function of the fractional parts with some singularity at zero.  To test the same scale of outliers as what is seen by the exponential type $B_{\sup}(x)$, we consider the logarithm function $-\log\{lx\}\in(0,+\infty)$.  One could use a different function $\Phi(\{lx\})$ with a logarithmic singularity at zero, in which case some of the constants would be replaced by integrals $\int_0^1\Phi(y)\,dy$.  In applications this would lead to tighter lower and upper bounds on the radius  in \autoref{sec:badly-approx} and \autoref{sec:well-approx}.

\begin{proposition}\label{prop:sum-log}
  For an irrational number $x\in\RR\setminus\QQ$ the limit
  \[
    C_{\sup}(x) \coloneqq \limsup_{L\geq 1} \frac{1}{L} \sum_{l=1}^L \bigl(-\log\{lx\}-\log\{-lx\}\bigr) \in [0,+\infty]
  \]
  is comparable\footnote{It appears plausible to tighten these bounds to $2+B_{\sup}(x)\leq C_{\sup}(x) \leq 2 + 4B_{\sup}(x)$.} to the exponential type $B_{\sup}(x)$:
  \[
    B_{\sup}(x) \leq C_{\sup}(x) \leq 4 + 4B_{\sup}(x) .
  \]
  In particular, $x$ is of finite exponential type if and only if $C_{\sup}(x)$ is finite.
\end{proposition}

\begin{proof}
The lower bound $B_{\sup}\leq C_{\sup}$ is immediate by keeping only the $l=L$ term in the sum and using that $-\log\{Lx\}-\log\{-Lx\}>\log(1/\dist(Lx,\ZZ))$.
We thus focus on the upper bound.
Denote $S^\epsilon_L\coloneqq\sum_{l=1}^L (-\log\{\epsilon lx\})$ for $\epsilon=\pm$ and $L\geq 1$.
Our strategy will be to bound these sums using the equidistribution of~$\{\pm lx\}$.

To avoid minor inconveniences later on, we restrict to $L\geq q_2$; this has no effect on the $L\to+\infty$ limit.  Note that $q_2 = a_2 a_1 + 1 \geq 2$.

Let $n\geq 2$ be such that $q_n\leq L<q_{n+1}$, and write
\begin{equation}
  L = M q_n + R , \qquad M\in\intset{1}{a_{n+1}} , \qquad R\in\intset{0}{q_n-1} ,
\end{equation}
with $R<q_{n-1}$ if $M=a_{n+1}$.
By Lemma~\ref{lem:equidistribution}, for the relevant range of values~$l$, the integer part $\varphi_{q_n,x}(\epsilon l)=\lfloor q_n\{\epsilon lx\}\rfloor$ only depends on~$\epsilon$ and the class $l\bmod q_n$.

For all values $l\in\intset{1}{L}$ with $\varphi_{q_n,x}(\epsilon l)\neq 0$, one has
\begin{equation}
  -\log\{\epsilon lx\} = \log q_n - \log\bigl(q_n\{\epsilon lx\}\bigr)
  \leq \log q_n - \log \varphi_{q_n,x}(\epsilon l) .
\end{equation}
For the remaining values, $\varphi_{q_n,x}(\epsilon l)=0$, namely $\{\epsilon lx\}\in[0,1/q_n)$.  Given that $q_n\geq q_2\geq 2$, this fractional part is equal to the distance $\dist(lx,\ZZ)$, which admits the lower bound $\dist(lx,\ZZ)\geq\dist(q_nx,\ZZ)>1/(2q_{n+1})$ from Lemma~\ref{lem:dist-greater} and~\eqref{estimate-quality-qnpn}.
Because $\varphi_{q_n,x}$ is bijective on $q_n$-element subintervals of $\pm\intset{1}{L}$ there are at most $1+\lfloor L/q_n\rfloor$ such values of~$l$.
Overall, we reach the following bound, using bijectivity of~$\varphi$ on relevant intervals to convert the sum of $\log\varphi_{q_n,x}(\epsilon l)$ to a sum of $\log j$,
\begin{equation}
  \begin{aligned}
    S^\epsilon_L
    & \leq (M + 1)\log(2q_{n+1}) + \sum_{1\leq l\leq L,\ \varphi_{q_n,x}(\epsilon l)\neq 0}\bigl(\log q_n - \log\varphi_{q_n,x}(\epsilon l)\bigr)
    \\
    & \leq (M+1) \log(2 q_{n+1}) + (M+1) \sum_{j=1}^{q_n-1} (\log q_n - \log j)
    \\
    & \leq (M+1) (\log(2 q_{n+1}) + q_n)
    \leq 2\Bigl(1+\frac{\log q_{n+1}}{q_n}+\frac{\log 2}{q_n}\Bigr) L ,
  \end{aligned}
\end{equation}
where in the next-to-last step we used that $q_n^{q_n-1}/(q_n-1)! = q_n^{q_n}/q_n!\leq e^{q_n}$ and in the last line we bounded $(M+1)q_n\leq 2L$.
By adding together the bounds on $S^+_L$ and~$S^-_L$ and passing to the lim sup, we conclude that
\begin{equation}
  C_{\sup}(x) \leq 4 + 4 B_{\sup}(x) .
  \qedhere
\end{equation}
\end{proof}

\section{Convergence for most \texorpdfstring{$b^2>0$}{b\texttwosuperior > 0}}
\label{sec:badly-approx}

\paragraph{Set-up.}

We now assume that $b^2>0$ is an irrational number that has a finite exponential type (Definition~\ref{def:expo-type}).
The lattice $\eps_1\ZZ+\eps_2\ZZ$ is dense in the line $\eps_1\RR=\eps_2\RR$.
The differences of Coulomb branch parameters $a_I-a_J$ for $1\leq I<J\leq N$, and the mass~$m$, obey the genericity assumption
\be
m \not\in \eps_2\RR , \qquad
a_I-a_J \not\in \eps_2\RR .
\ee
In this regime, we exhibit a positive radius~$R$, given in~\eqref{finite-expo-R-val} below, such that the absolute Nekrasov sum converges for $|\qinst|<R$.
We do not attempt to achieve the optimal value of~$R$.

\paragraph{Off-diagonal factors.}

As we do not care about the precise value of~$R$, we content ourselves with an easy bound for the off-diagonal factors $Z_{IJ}$, $I\neq J$ in~\eqref{Zstar-21}.
Specifically, we observe that for any $\alpha,\beta\in\CC$ with $\alpha\not\in\RR$, the curve traced by $y\mapsto 1-\beta/(\alpha+y)$ is bounded (it is in fact a circle) since $|\alpha+y|\geq|\Im\alpha|>0$ and the $y\to\pm\infty$ limit is the finite value~$1$.
We can thus define the finite number
\bel{Calphabeta-val}
C_{\alpha,\beta} \coloneqq \sup_{y\in\RR} \Bigl|1 - \frac{\beta}{\alpha+y}\Bigr|
= \frac{|\beta| + |\beta-2i\Im\alpha|}{2|\Im\alpha|}
\in (0, +\infty) ,
\ee
which satisfies $C_{\alpha,\beta}=C_{\kappa\alpha,\kappa\beta}$ for any real $\kappa\neq 0$, and is independent of $\Re\alpha$.
Then the auxiliary function $Z_{\lambda\mu}$ in~\eqref{Zstar-21} obeys
\be
|Z_{\lambda\mu}(a,m)| = \prod_{(i,j)\in\lambda} \Bigl|1-\frac{m}{a-\eps_1 (\mu'_j-i) + \eps_2 (\lambda_i-j)}\Bigr|
\leq (C_{a/\eps_2,m/\eps_2})^{|\lambda|} ,
\ee
and thus
\bel{ZIJ54}
\prod_{1\leq I\neq J\leq N} |Z_{IJ}|
\leq \Bigl( \sup_{1\leq I\leq N} \prod_{J\neq I}\prod_{\pm} C_{(a_I-a_J)/\eps_2,\pm m/\eps_2} \Bigr)^{|\vec{Y}|} .
\ee

\paragraph{A first majoration.}

We turn to diagonal factors~$Z_{II}$ in~\eqref{Zstar-21}.
Fix $1\leq I\leq N$ and denote $\lambda=Y_I$ for brevity, so
\be
Z_{II} = \!\! \prod_{(i,j)\in\lambda} \! \Bigl(1-\frac{m/\eps_2}{(1+\lambda_i-j) -b^2 (\lambda'_j-i)}\Bigr) \Bigl(1-\frac{m/\eps_2}{b^2 (1+\lambda'_j-i) - (\lambda_i-j)}\Bigr) .
\ee
Introduce a monotonically decreasing function $\Phi\colon(0,+\infty)\to(0,+\infty)$,
\be
\Phi(d) = \log\Bigl(\max_{x\in[-1/d,1/d]} |1 - x m/\eps_2|\Bigr) ,
\ee
which has the asymptotics $\Phi(d) = - \log d + \log|m/\eps_2| + O(d)$ as $d\to 0$ and $\Phi(d)\to 0$ as $d\to+\infty$.
In terms of this function,
\bel{logZII-57}
\log|Z_{II}| \leq \!\! \sum_{(i,j)\in\lambda} \! \Bigl( \Phi\bigl(|b^2(\lambda'_j-i) - (1+\lambda_i-j)|\bigr) + \Phi\bigl(|b^2 (1+\lambda'_j-i) - (\lambda_i-j)|\bigr) \Bigr) .
\ee

\paragraph{Organizing boxes by arm length.}

Let us determine which boxes have $|b^2(\lambda'_j-i) - (1+\lambda_i-j)|<1$.
For any integer $q\geq 0$, consider the boxes with $\lambda'_j-i=q$, namely boxes of the form $(\lambda'_j-q,j)$ for $1\leq j\leq\lambda_{q+1}$ (the upper bound arises from the condition $\lambda'_j-q\geq 1$).
For these boxes,
\be
\bigl|b^2(\lambda'_j-i) - (1+\lambda_i-j)\bigr|
= \bigl|b^2 q - (1+\lambda_i-j)\bigr| ,
\ee
which is equal to the fractional parts $\{b^2q\}$ and $\{-b^2q\}$ when $1+\lambda_i-j=\lfloor b^2q\rfloor$ and $\lceil b^2q\rceil$, respectively.  For other values of $1+\lambda_i-j$, it is greater than~$1$.
We denote
\beal
n_q & = \Card\bigl\{ (i,j) \bigm| \lambda'_j-i=q \text{ and } 1+\lambda_i-j = \lfloor b^2q\rfloor\bigr\} , \\
n_{-q} & = \Card\bigl\{ (i,j) \bigm| \lambda'_j-i=q \text{ and } 1+\lambda_i-j = \lceil b^2q\rceil\bigr\} .
\ee
Note that $n_0=n_{-0}=0$ since $1+\lambda_i-j\geq 1$ cannot be equal to~$0$.
All boxes counted by $n_{\pm q}$ have the form $(\lambda'_j-q,j)$ with $1\leq j\leq\lambda_{q+1}$ so we learn that $n_q+n_{-q}\leq\lambda_{q+1}$.
Thus,
\beal
\quad & \unquad \biggl( \sum_{(i,j)\in\lambda} \Phi\bigl(|b^2(\lambda'_j-i) - (1+\lambda_i-j)|\bigr) \biggr)
- |\lambda| \Phi(1)
\\
& \leq \sum_{q=1}^{\lambda'_1-1} \sum_{\pm} n_{\pm q} \Bigl( \Phi(\{\pm b^2q\}) - \Phi(1) \Bigr) ,
\ee
where $\Phi(1)$ has been subtracted from all terms for convenience, we used $\Phi(d)-\Phi(1)\leq 0$ for $d\geq 1$, and $q=0$ does not appear in the sum since $n_{\pm 0}=0$.  It proves useful to increase the upper bound to $\lambda'_1$ by noting that $n_{\pm q}=0$ for $q=\lambda'_1$ (and higher) since $n_q+n_{-q}\leq\lambda_{\lambda'_1+1}=0$.

The boxes counted by $n_q$ are all in different rows and different columns since the row or column number is determined from the other by $i=\lambda'_j-q$ and $j = 1+\lambda_i - \lfloor b^2q\rfloor$.  In addition, these maps from $i$ to $j$ and back are monotonically decreasing, so we can label the boxes $(i_1,j_1),(i_2,j_2),\dots$ with increasing values $i_1<i_2<\dots$ and decreasing $j_1>j_2>\dots$.  These are all positive integers, so we find $i_l\geq l$ and $j_l\geq n_q+1-l$ for $1\leq l\leq n_q$.
The partition~$\lambda$ thus contains the staircase partition with side length $n_q$ and with $n_q(n_q+1)/2$ boxes.\footnote{In fact, for each box $(i_l,j_l)$ we also know that $\lambda$ contains a hook with arm/leg lengths $q$ and $\lfloor b^2q\rfloor-1$, namely it contains $(i_l+q,j_l)$ and $(i_l,j_l+\lfloor b^2q\rfloor-1)$.  This leads to improved bounds on~$n_{\pm q}$, whose eventual effect on the convergence radius would be interesting to track.}
This implies $n_q \leq \sqrt{2|\lambda|}$.  Together with the same bounds on~$n_{-q}$ and the previous upper bound we get
\be
0\leq n_{\pm q} \leq \ell_{q+1} \coloneqq \min\bigl( \lambda_{q+1}, \sqrt{2|\lambda|} \bigr) .
\ee

The same logic applies to the second term in~\eqref{logZII-57}.  The integers defined for $q\geq 0$ by
\beal
n'_q & = \Card\bigl\{ (i,j) \bigm| 1+\lambda'_j-i=q \text{ and } \lambda_i-j = \lfloor b^2q\rfloor\bigr\} , \\
n'_{-q} & = \Card\bigl\{ (i,j) \bigm| 1+\lambda'_j-i=q \text{ and } \lambda_i-j = \lceil b^2q\rceil\bigr\}
\ee
obey $n'_0=n'_{-0}=0$ and
\be
0\leq n'_{\pm q} \leq \ell_q = \min\bigl(\lambda_q, \sqrt{2|\lambda|}\bigr) .
\ee
We reach
\beall{logZ-510}
\quad & \unquad \log|Z_{II}| - 2|\lambda| \Phi(1)
\leq \sum_{q=1}^{\lambda'_1} \sum_{\pm}
(n'_{\pm q} + n_{\pm q}) \Bigl( \Phi(\{\pm b^2q\}) - \Phi(1) \Bigr) .
\\
 & \leq \sum_{q=1}^{\lambda'_1} \sum_{\pm} \bigl( \ell_q + \ell_{q+1} \bigr)
 \Bigl( \Phi(\{\pm b^2q\}) - \Phi(1) \Bigr) .
\ee

\paragraph{Bounding the partial sums.}

Since the coefficients $\ell_q+\ell_{q+1}$ in~\eqref{logZ-510} are non-increasing with~$q$, the sum can usefully be rewritten in terms of
\be
S_q = \sum_{p=1}^q \sum_{\pm} \Bigl(\Phi(\{\pm pb^2\}) - \Phi(1)\Bigr) .
\ee
We now determine an upper bound for $S_q$.
Note that
\bel{Phi1-val}
\Phi(1) = \log\bigl(\max(|m/\eps_2-1|, |m/\eps_2+1|)\bigr) .
\ee
Thanks to the asymptotics near zero, $\Phi(d)+\log d$ is bounded for $0<d\leq 1$, and specifically
\beall{Cmeps2-val}
C_{m/\eps_2} \, & \! \coloneqq \sup_{d\in(0,1]} \bigl(\Phi(d)-\Phi(1)+\log d \bigr)
= \log\Bigl(\max_{d\in[0,1],y\in[-1,1]} |d - y m/\eps_2|\Bigr)-\Phi(1)
\\
& = \log\bigl(\max(|m/\eps_2|,|m/\eps_2-1|, |m/\eps_2+1|)\bigr)-\Phi(1)
\\
& = \max(0,\log|m/\eps_2|-\Phi(1)).
\ee
Thus,
\be
0 \leq \Phi(d) - \Phi(1) \leq - \log d + C_{m/\eps_2} , \qquad d\in(0,1] .
\ee
We deduce
\bel{Sqbound}
S_q \leq q \biggl( 2 C_{m/\eps_2} + \frac{1}{q} \sum_{p=1}^q \bigl(-\log\{pb^2\}-\log\{-pb^2\}\bigr) \biggr) .
\ee
Recall now the definition of $C_{\sup}(b^2)$ in Proposition~\ref{prop:sum-log} as a lim sup of the average of $-\log\{pb^2\}-\log\{-pb^2\}$.  This gives
\bel{Sqbound-limsup}
\limsup_{q\to+\infty} \frac{1}{q} S_q \leq 2 C_{m/\eps_2} + C_{\sup}(b^2) .
\ee
An immediate consequence is that for any $\delta>0$, the difference $S_q-q(2C_{m/\eps_2}+C_{\sup}(b^2)+\delta)$ is negative for large enough~$q$, hence has a maximum $K_{\delta,\eps_1,\eps_2,m}\in\RR$ for $q\geq 1$.  For any $\delta>0$ and $q\geq 1$,
\bel{Sqbound-delta}
S_q \leq q (2 C_{m/\eps_2} + C_{\sup}(b^2) + \delta) + K_{\delta,\eps_1,\eps_2,m} .
\ee
None of the objects in this inequality depend on the partition~$\lambda$ nor the Coulomb branch parameters~$a$.

\paragraph{Summing the partial sums.}

We express~\eqref{logZ-510} in terms of~$S_q$ and use the bound~\eqref{Sqbound} and~\eqref{Sqbound-delta} to get, for any $\delta>0$,
\beal
& \log|Z_{II}| - 2|\lambda| \Phi(1)
\leq \sum_{q=1}^{\lambda'_1} (\ell_q - \ell_{q+2}) S_q
\\
& \leq (2 C_{m/\eps_2} + C_{\sup}(b^2) + \delta) \sum_{q=1}^{\lambda'_1} q(\ell_q-\ell_{q+2}) + K_{\delta,\eps_1,\eps_2,m} \sum_{q=1}^{\lambda'_1} (\ell_q-\ell_{q+2}) .
\ee
Telescoping both sums yields a simpler expression (simplified further by doubling a term~$\ell_1$), which we then bound by replacing $\ell_q=\min(\sqrt{2|\lambda|},\lambda_q)$ by $\sqrt{2|\lambda|}$ or $\lambda_q$ depending on the term,
\beall{ZII519}
\log|Z_{II}|
& \leq 2|\lambda| \Phi(1) + 2 \bigl(2 C_{m/\eps_2} + C_{\sup}(b^2) + \delta\bigr) \sum_{q=1}^{\lambda'_1} \ell_q + K_{\delta,\eps_1,\eps_2,m} (\ell_1 + \ell_2)
\\
& \leq 2\bigl(\Phi(1) + 2 C_{m/\eps_2} + C_{\sup}(b^2) + \delta\bigr) |\lambda| + 2 K_{\delta,\eps_1,\eps_2,m} \sqrt{2|\lambda|} .
\ee

\paragraph{Collecting off-diagonal and diagonal factors.}

We combine~\eqref{ZIJ54} and~\eqref{ZII519} to get
\beall{logZY524}
\log(|Z_{\vec{Y}}|)
& \leq \sum_{1\leq I,J\leq N} \log|Z_{IJ}|
\\
& \leq |\vec{Y}| \log\Bigl( \sup_{1\leq I\leq N} \prod_{J\neq I}\prod_{\pm} C_{(a_I-a_J)/\eps_2,\pm m/\eps_2} \Bigr) \\
& \quad + 2|\vec{Y}|\bigl(\Phi(1) + 2 C_{m/\eps_2} + C_{\sup}(b^2) + \delta\bigr) + O(|\vec{Y}|^{1/2})
\ee
for all $\vec{Y}$, for any given $\delta>0$.
Since this holds for all $\delta>0$, the lim inf~\eqref{Rdef-Z} that gives $\Rabs$ is bounded below and one gets
\beall{finite-expo-R-val}
\Rabs
& \geq A(m,a;\eps_1,\eps_2) e^{-2C_{\sup}(b^2)} ,
\\
A(m,a;\eps_1,\eps_2) \, & \! \coloneqq e^{-2\Phi(1) - 4 C_{m/\eps_2}} \inf_{1\leq I\leq N} \prod_{J\neq I}\prod_{\pm} \bigl( C_{(a_I-a_J)/\eps_2,\pm m/\eps_2} \bigr)^{-1} ,
\ee
where the factors in $A(m,a;\eps_1,\eps_2)$ are defined in \eqref{Calphabeta-val}, \eqref{Phi1-val}, \eqref{Cmeps2-val}, and all are continuous and piecewise smooth, in contrast to the densely discontinuous $C_{\sup}(b^2)$.

\paragraph{Some remarks.}

By Proposition~\ref{prop:sum-log}, the bound can be further weakened to
\be
\Rabs \geq A(m,a;\eps_1,\eps_2) e^{-8-8B_{\sup}(b^2)} .
\ee
The same bound with $b\to 1/b$ also holds.
Under this transformation, the $C_{a/\eps_2,m/\eps_2}$ factors in $A(m,a;\eps_1,\eps_2)$ are invariant, but the exponential prefactor is not, and most importantly $B_{\sup}(b^2)$ is mapped to $B_{\sup}(1/b^2) = B_{\sup}(b^2)/b^2$ by Proposition~\ref{prop:B-inverse}.
Thus, up to redefining $A$ to a new function~$A_1$ to absorb some factors, one gets the lower bound in~\eqref{Rabs-lower-bound-thm2star}.

Details of the function $A(m,a;\eps_1,\eps_2)$ are artifacts of loose bounds along the way.  However, some dependence on Coulomb branch parameters appears to be necessary: for $\vec{Y}=(\lambda,\lambda,\dots,\lambda)$ with $\lambda$ the partition defined for some~$L$ by
\be
\lambda_i = 1 + \bigfloor{b^2 (L-i)} , \qquad i = 1,\dots,L ,
\ee
namely $\lambda$ close to a triangle with (irrational) slope~$b^2$, one can check that for any box $(i,j)\in\lambda$, its box content lies in a small segment of the line $\eps_2\RR$ since
\beal
\quad & \unquad \bigl( -(\lambda'_j-i)\eps_1+(\lambda_i-j)\eps_2 \bigr) / \eps_2
\\
& = - \bigl(L-i - \bigceil{b^{-2} (j - 1)}\bigr) b^2
+ \bigfloor{b^2 (L-i)} + 1 - j
\\
& = - \bigl\{ b^2 (L-i) \bigr\} + b^2 \bigl\{ b^{-2} (1-j) \bigr\}
\\
& \in (-1,b^2) .
\ee
All off-diagonal factors are then close to $C_{a/\eps_2,m/\eps_2}$, and contribute non-trivially to the leading term in $\frac{1}{|\vec{Y}|} \log|Z_{\vec{Y}}|$ in a way similar to how they appear in~\eqref{logZY524}.

\section{Radius upper bounds for \texorpdfstring{$b^2>0$}{b\texttwosuperior > 0}}
\label{sec:divergence}

This section is devoted to divergence properties of the absolute sum.
In the rational $b^2>0$ case, \autoref{sec:divergence-rational} determines that the Nekrasov sum is simply ill-defined due to some singular terms.
For other values of~$b^2$, \autoref{sec:Rabs_le_1} proved that $\Rabs\leq 1$ under our genericity assumptions.
\autoref{sec:well-approx} provides better upper bounds for irrational $b^2>0$, depending on its exponential type $B_{\sup}(b^2)$.
This culminates in a proof that $\Rabs=0$ for super-exponentially well approximable positive irrational~$b^2$, namely when $B_{\sup}(b^2)=+\infty$.

\subsection{Rational case: ill-defined terms}
\label{sec:divergence-rational}

We consider the positive rational case $b^2=p/q$ where $p,q>0$ are coprime integers.
Equivalently, we assume that the lattice $\eps_1\ZZ+\eps_2\ZZ=\veps\ZZ$ for some~$\veps$, so that
\be
\eps_1 = p\veps , \qquad \eps_2 = q\veps , \qquad \veps = \gcd(\eps_1,\eps_2) .
\ee
Under the genericity assumption
\be
m\not\in \eps_1\ZZ+\eps_2\ZZ = \veps\ZZ
\ee
(the genericity assumption on Coulomb branch parameters is unnecessary here),
we exhibit terms in the sum over~$\vec{Y}$ that have poles as $b^2\to p/q$.  This implies that the sum over tuples is \emph{ill-defined}\footnote{Note that after cancellations between contributions of different tuples with $|\vec{Y}|=k$ the partition function itself may just have a removable singularity at this value of~$b^2$, but the techniques used in this paper are not suited to answering that.} for that value of~$b^2$.

Consider the term $\vec{Y}=(\lambda,\emptyset,\dots,\emptyset)$, so that the expression~\eqref{Zstar-21} simplifies as follows (after mapping $j\to\lambda_i+1-j$ in the product defining~$Z_{1J}$ for $J\neq 1$)
\begin{align}\label{Zlambda00}
& Z_{\vec{Y}=(\lambda,\emptyset,\dots)} = \prod_{J=1}^N Z_{1J} ,
\\
\nonumber
& Z_{1J} \overset{\mathclap{J\neq 1}}{=} \prod_{(i,j)\in\lambda} \frac{a_1-a_J+i\eps_1 + j\eps_2-m}{a_1-a_J+i\eps_1 + j\eps_2} \, \frac{a_1-a_J+(i-1)\eps_1 + (j-1)\eps_2+m}{a_1-a_J+(i-1)\eps_1 + (j-1)\eps_2} ,
\\
\nonumber
& Z_{11} = \!\! \prod_{(i,j)\in\lambda} \!\! \frac{-\eps_1 (\lambda'_j-i) + \eps_2 (1+\lambda_i-j)-m}{-\eps_1 (\lambda'_j-i) + \eps_2 (1+\lambda_i-j)} \,
\frac{\eps_1 (1+\lambda'_j-i) - \eps_2 (\lambda_i-j)-m}{\eps_1 (1+\lambda'_j-i) - \eps_2 (\lambda_i-j)} .
\end{align}

For any partition $\lambda$ with $\lambda_1=p$ and $\lambda'_1=q+1$, the box $(i,j)=(1,1)$ contributes the factor $1/(p\eps_2-q\eps_1)$, which has a pole at $b^2=p/q$.
Generically, this is enough to ensure that $Z_{\vec{Y}}$ blows up as $b^2\to p/q$, but for fine-tuned Coulomb branch or mass parameters this singularity can be cancelled by factors in the numerator.

We have assumed the mass to be generic, that is, $m\not\in \eps_1\ZZ+\eps_2\ZZ$, such that the numerator of~$Z_{11}$ remains non-zero.\footnote{In \autoref{app:comment-pole-cancel}, we explore some special values of~$m$ where poles might be avoided.} Regarding the numerator of $Z_{1J}$ for $J\neq 1$, observe that $i\eps_1+j\eps_2=a_J-a_1+m$ or $a_J-a_1-m+\eps_1+\eps_2$ has finitely many solutions $(i,j)$ since $b^2=\eps_1/\eps_2>0$, so even for fine-tuned Coulomb branch parameters, only a bounded number of numerator factors may vanish as $b^2\to p/q$.
One can get arbitrarily many denominator factors to vanish by the following choice of partition~$\lambda$.  This ensures that $Z_{\vec{Y}}$ blows up as $b^2\to p/q$.

For $p>1$ one can take a staircase partition defined by a parameter $K\geq 1$ in addition to $(p,q)$, with
\be
\lambda_{kq+1} = \dots = \lambda_{kq+q} = (K+1-k)(p-1) , \qquad 0\leq k\leq K .
\ee
For this partition, the boxes $(i,j)=(kq,(K+1-k)(p-1))$ for $1\leq k\leq K$ have
\beal
\lambda_i = (K+2-k)(p-1) , \quad
\lambda'_j = (k+1)q ,
\\
-\eps_1 (\lambda'_j-i) + \eps_2 (1+\lambda_i-j)
= -\eps_1 q + \eps_2 p = 0 .
\ee
This term~$\vec{Y}$ thus features $K$~vanishing denominator factors, and $K$~can be made arbitrarily large.
In the special case $p=1$, the same construction works with $(p,q)$ replaced by $(2p,2q)$ or any higher multiple thereof.

\subsection{Bound by the exponential type}
\label{sec:well-approx}

\paragraph{Set-up.}

We now assume that $b^2$ is a positive irrational number and that the mass obeys the genericity assumption
\be
m\not\in\overline{\eps_1\ZZ+\eps_2\ZZ} = \eps_2\RR .
\ee
(The genericity assumption on Coulomb branch parameters is unnecessary here.)
We show an upper bound on the absolute convergence radius~$\Rabs$:
\be
\Rabs \leq e^{-B_{\sup}(b^2)/(1+b^2)} A_2(m,a;\eps_1,\eps_2)
\ee
that depends on the exponential type $B_{\sup}(b^2)$ of~$b^2$ defined in Definition~\ref{def:expo-type}, and on a continuous function $A_2(m,a;\eps_1,\eps_2)$ of the parameters.
Observe that $B_{\sup}(b^2)/(1+b^2)$ is invariant under $b\to 1/b$ by Proposition~\ref{prop:B-inverse}.
For super-exponentially well approximable $b^2>0$, namely numbers with $B_{\sup}(b^2)=+\infty$, one gets $\Rabs=0$, which means that the sum over partitions does not converge absolutely for any~$\qinst$.

The upper bound on~$\Rabs$ is obtained by exhibiting a (sub)sequence of tuples~$\vec{Y}$ and lower bounds on $\log|Z_{\vec{Y}}|$ that are linear in~$|\vec{Y}|$, of the form
\be
\log|Z_{\vec{Y}}| \geq \Bigl( \frac{B}{1+b^2} - \log A_2 + o(1) \Bigr) |\vec{Y}| ,
\qquad |\vec{Y}| \to +\infty .
\ee

\paragraph{The choice of term.}

As in the rational $b^2>0$ case, we shall consider terms for which $\vec{Y}$ has a single non-empty component.
Before we start, it is convenient to reorder without loss of generality the Coulomb branch parameters~$a_I$ such that $\alpha_1\leq\alpha_2\leq\dots\leq\alpha_N$, with
\bel{alphaI-def}
\alpha_I \coloneqq \frac{\Re\eta}{\Im\eta} \Im\Bigl(\frac{a_I}{\eps_2}\Bigr) - \Re\Bigl(\frac{a_I}{\eps_2}\Bigr) , \qquad
\eta \coloneqq \frac{2m-\eps_1-\eps_2}{2\eps_2} ,
\ee
where $\Im\eta\neq 0$ thanks to the genericity of~$m$.
This ordering prevents some numerators from vanishing in upcoming expressions.

We also parametrize $m/\eps_2$ and $\eta$ in polar coordinates $\rho_m,\rho_\eta>0$ and $\theta_m,\theta_\eta\in(0,\pi)\cup(\pi,2\pi)$,
\be
m/\eps_2 = \rho_m e^{i\theta_m} , \qquad
\eta = \rho_\eta e^{i\theta_\eta} .
\ee

\paragraph{A small denominator.}

Consider the term with $\vec{Y}=(\lambda,\emptyset,\dots,\emptyset)$, whose contribution was already given in~\eqref{Zlambda00}.
In particular,
\bel{Z11-expr}
Z_{11}
= \!\! \prod_{(i,j)\in\lambda} \!\! \Bigl(1 - \frac{m/\eps_2}{-b^2 (\lambda'_j-i) + (1+\lambda_i-j)} \Bigr)
\Bigl( 1 - \frac{m/\eps_2}{b^2 (1+\lambda'_j-i) - (\lambda_i-j)} \Bigr) .
\ee
Let us make one of the factors large.
For some $q\geq 1$ to be chosen later, let $p=\lfloor qb^2+1/2\rfloor$ be the integer closest to $qb^2$, and select $\lambda$ to be the smallest partition with $\lambda_1=p$ and $\lambda'_1=q+1$, namely $\lambda=(p,1^q)$.
The first factor in~\eqref{Z11-expr} for $(i,j)=(1,1)$ is then
\be
1 - \frac{m/\eps_2}{-b^2 (\lambda'_j-i) + (1+\lambda_i-j)} \Bigr|_{i=j=1}
= 1 \pm \frac{m/\eps_2}{\dist(q b^2,\ZZ)} .
\ee
If $qb^2$ is close to an integer, this becomes large.

\paragraph{Other diagonal factors.}

The remaining $2|\lambda|-1=2p+2q-1$ factors in~$Z_{11}$ are bounded below as follows.
In terms of $m/\eps_2=\rho e^{i\theta}$, the genericity assumption means that $1$~is at a finite distance from the line $e^{i\theta}\RR$, which is explicited in the lower bound
\be
\bigl| 1 - y(m/\eps_2) \bigr|
= \bigl| e^{-i\theta} - \rho y\bigr|
\geq |\Im e^{-i\theta}| = \abs{\sin\theta} > 0 ,
\qquad y\in\RR .
\ee
Thus~\eqref{Z11-expr} is bounded below by
\be
|Z_{11}|
\geq \abs{\sin\theta}^{2p+2q-1} \Bigl|1 \pm \frac{m/\eps_2}{\dist(q b^2,\ZZ)}\Bigr| .
\ee

\paragraph{Off-diagonal factors.}

Next, we consider $Z_{1J}$ for $J\neq 1$, explicited in~\eqref{Zlambda00}.
Recall $\eta = m/\eps_2 - (1+b^2)/2$ and the definition of~$\alpha_I$ in~\eqref{alphaI-def}.
The two numerator factors in~$Z_{1J}$ for $J\neq 1$, divided by~$\eps_2$ for convenience, take the form
\beall{a1eps2-aJeps2}
\frac{a_1}{\eps_2} - \frac{a_J}{\eps_2} \pm \eta + \Bigl(i-\frac{1}{2}\Bigr)b^2 + \Bigl(j-\frac{1}{2}\Bigr)
& \in \frac{a_1}{\eps_2} - \frac{a_J}{\eps_2} + \eta \RR + \Bigl[\frac{1+b^2}{2},+\infty\Bigr) \\
& = \eta \RR + \Bigl[\alpha_J-\alpha_1+\frac{1+b^2}{2},+\infty\Bigr) ,
\ee
where the last step relies on $a_I/\eps_2 = - \alpha_I + (\Im(a_I/\eps_2)/\Im\eta) \eta\in-\alpha_I+\RR\eta$.
Thanks to the ordering, $\alpha_J-\alpha_1+(1+b^2)/2>0$, so the last set in~\eqref{a1eps2-aJeps2} does not contain~$0$.
In terms of $\eta=\rho_\eta e^{i\theta_\eta}$, one has the lower bound
\be
\min_{z\in \eta\RR+[\alpha_J-\alpha_1+(1+b^2)/2,+\infty)} |z|
= \Bigl(\alpha_J-\alpha_1+\frac{1+b^2}{2}\Bigr)\abs{\sin\theta_\eta} ,
\ee
which leads to a lower bound on~$Z_{1J}$, conveniently expressed as (for $J\neq 1$)
\beal
|Z_{1J}| & = \prod_{(i,j)\in\lambda} \prod_{\pm} \biggl|1 + \frac{m/\eps_2}{a_1/\eps_2-a_J/\eps_2\pm\eta+(i-1/2)b^2+(j-1/2)} \biggr|^{-1} \\
& \geq \biggl(1 + \frac{|m/\eps_2|}{(\alpha_J-\alpha_1+(1+b^2)/2)\abs{\sin\theta_\eta}} \biggr)^{-2(p+q)} .
\ee

\paragraph{Concluding.}

By~\eqref{Rdef-Z}, the absolute convergence radius is a lim inf over all tuples, hence is bounded above by the lim inf over tuples $\vec{Y}=(\lambda,\emptyset,\dots,\emptyset)$ with $\lambda=(p,1^q)$ that we are considering now.  Thus,
\beal
\Rabs & \leq \liminf_{q\to+\infty} |Z_{((p,1^q),\emptyset,\dots)}|^{-1/(p+q)}
= \liminf_{q\to+\infty} \Bigl( |Z_{11}|^{-1/(p+q)} \prod_{J\neq 1} |Z_{1J}|^{-1/(p+q)} \Bigr)
\\
& \leq A_2 \liminf_{q\to+\infty} \Bigl|1 \pm \frac{m/\eps_2}{\dist(q b^2,\ZZ)}\Bigr|^{-1/(p+q)} .
\ee
Using that $q(1+b^2)\in [p+q-1/2, p+q+1/2]$ by construction of~$p$, we note that $p+q=q(1+b^2)+O(1)$ as $q\to+\infty$. The effects of additive and multiplicative constants $1$ and $m/\eps_2$ drop out when raised to a power $-1/(p+q)\to 0$, so
\be
\Rabs \leq A_2 \Bigl( \liminf_{q\to+\infty} \Bigl( \dist(qb^2,\ZZ)^{1/q} \Bigr) \Bigr)^{1/(1+b^2)}
= A_2 e^{-B_{\sup}(b^2)/(1+b^2)} .
\ee

\section*{Acknowledgements}

The author thanks Sylvain Lacroix for stimulating discussions, as well as Fabrizio Del Monte, Harini Desiraju, Alba Grassi, and Vincent Vargas for organizing the ``Conformal field theory 3 ways: integrable, probabilistic, and supersymmetric'' SwissMAP conference in January 2024, which reignited his interest into these topics.

\appendix

\section{Instanton partition functions}
\label{sec:instanton-telescopic}

The $k$-instanton contribution to the Nekrasov partition function of the $U(N)$ theories of interest to us takes the form of a sum of residues of a matrix model~\cite{hep-th/0206161}.  Each term is a finite product/ratio of $O(k^2+kN)$ factors that are linear in the various parameters $m,a_I,\eps_1,\eps_2$.
Here, we briefly explain some well-known rewritings of these products in which one cancels a multitude of factors that are equal in the numerator and denominator, and reduces the products to have only $O(kN)$ factors.
For clarity, and reuse in later exploration of 4d $\Nsusy=2$ SQCD and 5d $\Nsusy=1$ instanton partition functions, we express the main simplification steps in terms of an abstract function~$f$, specialized to the 4d $\Nsusy=2^*$ theory afterwards.

\paragraph{Telescopic products from matrix model residue sums.}

An important building block of instanton partition functions is the product
\beall{telescope-1}
Z_{\vec{Y}}[f] & = \biggl(
\sideset{}'\prod_{1\leq\alpha\leq k,1\leq J\leq N} \frac{1}{f(\phi_\alpha-a_J+\eps_1+\eps_2)}
\sideset{}'\prod_{1\leq\beta\leq k,1\leq I\leq N} \frac{1}{f(a_I-\phi_\beta)} \\
& \qquad \sideset{}'\prod_{1\leq \alpha,\beta\leq k} \frac{f(\phi_\alpha-\phi_\beta)f(\phi_\alpha-\phi_\beta+\eps_1+\eps_2)}{f(\phi_\alpha-\phi_\beta+\eps_1)f(\phi_\alpha-\phi_\beta+\eps_2)}
\biggr)_{\phi=\phi(\vec{Y})}
\ee
for some function~$f$ with $f(0)=0$, where $\phi=\phi(\vec{Y})$ is given (with an arbitrary order of the entries~$\phi_\alpha$) by
\bel{phialpha}
\{\phi_\alpha\} = \bigl\{a_I+(i-1)\eps_1+(j-1)\eps_2\bigm| 1\leq I\leq N , (i,j)\in Y_I\bigr\} ,
\ee
and the primes denote the omission of all factors $f(0)=0$.
Alternatively, one can shift all $f(x)$ to $f(\delta+x)$ for some small $\delta$, and consider the leading term as $\delta\to 0$.
Since this appendix is only meant to outline the rewriting procedure, we will omit the primes in the products below.

Collecting together the factors of the form $f(a_I-a_J+p\eps_1+q\eps_2)\eqqcolon f_{IJ}(p,q)$ for fixed $I,J$, we find
\beal
Z_{\vec{Y}}[f] & = \prod_{1\leq I,J\leq N} Z_{Y_I Y_J}[f_{IJ}] ,
\\
Z_{\lambda\mu}[g]
\, & \! \coloneqq \prod_{(i,j)\in\lambda} \frac{1}{g(i,j)}
\prod_{(s,t)\in\mu} \frac{1}{g(1-s,1-t)} \\
& \prod_{(i,j)\in\lambda,(s,t)\in\mu} \frac{g(i-s,j-t)g(1+i-s,1+j-t)}{g(1+i-s,j-t)g(i-s,1+j-t)} .
\ee
Observe the symmetry $Z_{\lambda\mu}[g]=Z_{\mu\lambda}[\tilde{g}]$ where $\tilde{g}(i,j)=g(1-i,1-j)$.
We separate $Z_{\lambda\mu}[g]$ into factors $g(p,q)$ with $q\geq 1$ and those with $q\leq 0$, which are related to the first by symmetry,
\beal
Z_{\lambda\mu}[g] & = Z^>_{\lambda\mu}[g] Z^\leq_{\lambda\mu}[g] ,
\qquad\qquad Z^\leq_{\lambda\mu}[g] = Z^>_{\mu\lambda}[\tilde{g}] ,
\\
Z^>_{\lambda\mu}[g] & = \! \prod_{(i,j)\in\lambda} \biggl( \frac{1}{g(i,j)}
\prod_{\substack{1\leq t<j\\1\leq s\leq\mu'_t}} \frac{g(i-s,j-t)}{g(1+i-s,j-t)}
\prod_{\substack{1\leq t\leq j\\1\leq s\leq\mu'_t}} \frac{g(1+i-s,1+j-t)}{g(i-s,1+j-t)} \biggr)
\\[-6pt]
\ee
with the convention that $\mu'_t=0$ for $t>\mu_1$.

We now evaluate $Z^>_{\lambda\mu}[g]$.
We change the first factor $g(i,j)$ to $g(i,1+\lambda_i-j)$ by using that the interval $1\leq j\leq\lambda_i$ is invariant under this change.
We also perform the telescopic products over~$s$,
\beal
& Z^>_{\lambda\mu}[g] = \!\! \prod_{(i,j)\in\lambda} \biggl( \frac{1}{g(i,1+\lambda_i-j)}
\prod_{1\leq t<j} \!\! \frac{g(i-\mu'_t,j-t)}{g(i,j-t)}
\prod_{1\leq t\leq j} \frac{g(i,1+j-t)}{g(i-\mu'_t,1+j-t)} \biggr)
\\
& = \!\! \prod_{(i,j)\in\lambda} \frac{1}{g(i,1+\lambda_i-j)}
\prod_{(i,t)\in\lambda} \!\! \biggl( \prod_{t<j\leq\lambda_i} \!\! \frac{g(i-\mu'_t,j-t)}{g(i,j-t)} \!
\prod_{t\leq j\leq\lambda_i} \frac{g(i,1+j-t)}{g(i-\mu'_t,1+j-t)} \biggr)
\ee
where in the second line we have simply reordered the products over $i,j,t$ into products over $i,t,j$ with suitable bounds, observing that $(i,j)\in\lambda$ and $1\leq t\leq j$ implies $(i,t)\in\lambda$.
The last two products cancel, apart from the $j=\lambda_i$ term in the last product.
This ratio $g(i,1+\lambda_i-t)/g(i-\mu'_t,1+\lambda_i-t)$ finally combines with the first factor and yields
\be
Z^>_{\lambda\mu}[g] = \prod_{(i,j)\in\lambda} \frac{1}{g(i-\mu'_j,1+\lambda_i-j)} .
\ee
Altogether
\be
Z_{\lambda\mu}[g]
= \prod_{(i,j)\in\lambda} \frac{1}{g(i-\mu'_j,1+\lambda_i-j)}
\prod_{(i,j)\in\mu} \frac{1}{g(1+\lambda'_j-i,j-\mu_i)} .
\ee

\paragraph{Telescopic products for \texorpdfstring{$\Nsusy=2^*$}{N=2*}.}

The $k$-instanton contribution to the Nekrasov partition function takes the form of a sum of residues of a matrix model~\cite{hep-th/0206161},
\beall{Zkn2star}
Z_k & = \sum_{|\vec{Y}|=k} \biggl(
\sideset{}'\prod_{1\leq\alpha\leq k,1\leq J\leq N} \Bigl(1 - \frac{m}{\phi_\alpha-a_J+\eps_1+\eps_2}\Bigr)
\sideset{}'\prod_{1\leq\beta\leq k,1\leq I\leq N} \Bigl(1 - \frac{m}{a_I-\phi_\beta}\Bigr) \\
& \qquad \sideset{}'\prod_{1\leq \alpha,\beta\leq k}
\Bigl( \frac{1-m/(\phi_\alpha-\phi_\beta+\eps_1)}{1-m/(\phi_\alpha-\phi_\beta)} \Bigr)
\Bigl( \frac{1-m/(\phi_\alpha-\phi_\beta+\eps_2)}{1-m/(\phi_\alpha-\phi_\beta+\eps_1+\eps_2)} \Bigr)
\biggr)_{\phi=\phi(\vec{Y})} ,
\ee
with the same notation $\prod'$ and $\phi=\phi(\vec{Y})$ as in~\eqref{phialpha}.
Specifically, the summand in~\eqref{Zkn2star} coincides with the general form~\eqref{telescope-1} for $f(x) = x / (x-m)$, namely $1/f(x)=1-m/x$.
Thus, the sum $\Zinststar = \sum_{k\geq 0} \qinst^k Z_k$ is
\beall{tele-2star-final}
\Zinststar \! & = \sum_{\vec{Y}} \qinst^{|\vec{Y}|} \!\! \prod_{1\leq I,J\leq N} \biggl(
\prod_{(i,j)\in Y_I} \Bigl( 1 - \frac{m}{a_I-a_J+(i-Y'_{Jj})\eps_1+(1+Y_{Ii}-j)\eps_2} \Bigr) \\
& \qquad\qquad\qquad \prod_{(i,j)\in Y_J} \Bigl( 1 - \frac{m}{a_I-a_J+(1+Y'_{Ij}-i)\eps_1+(j-Y_{Ji})\eps_2} \Bigr) \biggr) .
\ee
To get the form~\eqref{Zstar-21} used in the main text, we further swap $I,J$ in the second set of factors.

\section{Absolute convergence radius}
\label{sec:absolute-convergence-radius}

Theorem~\ref{thm:2star} involves the notion of absolute convergence radius, and in the proof we rely on a formula~\eqref{Rdef-Z} as a limit inferior in terms of the coefficients of the series.  This appendix explores the notion and proves the formula.
This will allow us to eliminate unimportant factors and shifts in the main text.  The only subtlety compared to a standard complex analysis course is that for each power~$k$ there is a sum over tuples of partitions with $|\vec{Y}|=k$, whose number of terms thankfully grows subexponentially.

\begin{definition}\label{def:convrad}
  Given a collection of coefficients $A_{\vec{Y}}\in\CC$ labeled by $N$-tuples of partitions $\vec{Y}=(Y_1,\dots,Y_N)$, the \emph{absolute convergence radius}
  \bel{Rdef}
  R_{|A|} = \liminf_{\vec{Y}} |A_{\vec{Y}}|^{-1/|\vec{Y}|} \in [0,+\infty]
  \ee
  is the convergence radius of the series
  \bel{qAseries}
  \sum_{k\geq 0} q^k \sum_{|\vec{Y}|=k} \bigl|A_{\vec{Y}}\bigr| .
  \ee
\end{definition}

\begin{proof}[Proof of the formula for $R_{|A|}$]
  The limit inferior is the smallest accumulation point of the infinite set $|A_{\vec{Y}}|^{-1/|\vec{Y}|}$.
  This has two consequences.
  Firstly, for any (real) $q>R_{|A|}$, there are infinitely many $\vec{Y}$ with $|A_{\vec{Y}}|^{-1/|\vec{Y}|}<q$, hence infinitely many terms greater than~$1$ in the series~\eqref{qAseries}, which diverges as a consequence.
  Secondly, for any $0\leq r<R_{|A|}$ there are finitely many $\vec{Y}$ for which $|A_{\vec{Y}}|\geq r^{-|\vec{Y}|}$, hence there exists a constant $C>0$ such that $|A_{\vec{Y}}|\leq Cr^{-|\vec{Y}|}$.  Then, for any $q$ with $|q|<r$, one has
  \beal
  \biggl| \sum_{k\geq 0} q^k \sum_{|\vec{Y}|=k} \bigl|A_{\vec{Y}}\bigr| \biggr|
  & \leq C \sum_{k\geq 0} |q/r|^k \sum_{|\vec{Y}|=k} 1
  \\
  &
  = C \prod_{I=1}^N \biggl( \sum_{Y_I} |q/r|^{|Y_I|} \biggr)
  = C \biggl(\prod_{j\geq 1} (1 - |q/r|^j)^{-1}\biggr)^N ,
  \ee
  where the last equality is the standard expression of the generating series of partition numbers, which has radius of convergence~$1$.
\end{proof}

\section{Exceptional pole cancellation}
\label{app:comment-pole-cancel}

This appendix prolongs the investigation of positive rational $b^2=p/q>0$ with $p,q$ positive coprime integers.
For generic masses, \autoref{sec:divergence-rational} exhibited terms that have poles as $b^2\to p/q$.
For non-generic masses the pole may be cancelled, depending on precisely what partition~$\lambda$ is considered.
Consider first $\lambda$ to be the smallest partition with $\lambda_1=p$ and $\lambda'_1=q+1$, namely $\lambda=(p,1^q)$.  Then the numerator factors vanish precisely when the mass belongs to a finite subset of the lattice $\eps_1\ZZ+\eps_2\ZZ$,
\beall{mset-24}
m & \in \{0\}\cup\{\eps_1+\eps_2\} \cup \{k\eps_1\mid 1\leq k\leq q\} \cup \{k\eps_2\mid 1\leq k\leq p\} \\
& \quad \cup \{\eps_1+\eps_2-k\eps_1\mid 1\leq k\leq q\} \cup \{\eps_1+\eps_2-k\eps_2\mid 1\leq k\leq p\} .
\ee
For those masses, the specific term considered here is in fact not singular (has a removable singularity) at $b^2=p/q$.
There exist many other terms $Z_{\vec{Y}}$ of course, and the precise question to answer for a given fine-tuned mass is whether this mass avoids poles at $b^2=p/q$ for all tuples~$\vec{Y}$.

\paragraph{An example and the resulting conjecture.}

A full answer is beyond the scope of this paper, but let us discuss the case $b^2=2/5$ as an example, and write $\eps_1=2\veps$, $\eps_2=5\veps$.  Then the set~\eqref{mset-24} reads
\be
m/\veps \in \{-3,-1,0,1,2,3,4,5,6,7,8,10\} .
\ee
Consider now the partition $\lambda=(2,2,2,1,1,1)$.  It still has a singular factor, but now the numerators vanish for a different (smaller) set of masses, due to a change in the various hooks:
\be
m/\veps \in \{0,1,2,3,4,5,6,7\} .
\ee
This example (among others) suggests that the relevant set of masses is the intersection of~\eqref{mset-24} and of the segment $[0,1+b^2]\eps_2$.  We state this as the following conjecture.

\begin{conjecture}\label{conj:2star-rational}
  The value $b^2=p/q$ is a removable singularity (or smooth point) of all terms $Z_{\vec{Y}}$ if $m$ or $\eps_1+\eps_2-m$ is in $\{k\eps_1\mid 0\leq k\leq \lceil q/p\rceil\}\cup\{k\eps_2\mid 0\leq k\leq \lceil p/q\rceil\}$.
  For other values of~$m$, $b^2=p/q$ is a pole of at least one term~$Z_{\vec{Y}}$.
\end{conjecture}

The necessity of these conditions on~$m$ might be proven by choosing $\lambda$ to be an approximation of a triangular shape with slope $p/q$, as this eliminates hooks with large disparities of the arm and leg lengths, thus prevents large/small masses from cancelling the pole.
The sufficiency (namely that poles disappear for such masses) might be proven by showing that for any partition, for each hook whose arm/leg lengths have the ratio $p/q$ (up to shifts by~$1$) responsible for a pole at $b^2=p/q$, there is a hook that provides the necessary zero.
Off-diagonal factors $Z_{IJ}$, $1\leq I\neq J\leq N$ play no role in this discussion under our genericity assumption $a_I-a_J\not\in\eps_1\ZZ+\eps_2\ZZ$.

Incidentally, for $m=0$ and $m=\eps_1+\eps_2$ each instanton configuration $\vec{Y}$ contributes exactly~$1$ because each $Z_{II}=1$ and $Z_{IJ}Z_{JI}=1$.
For more general masses suggested by Conjecture~\ref{conj:2star-rational}, since the presence of zeros for particular masses does not rely on a fine-tuning of~$b^2$, it should be interesting to consider whether the instanton partition function simplifies without assumption on~$b^2$.

Despite appearances, the AGT dictionary \emph{does not} map these conformal blocks to minimal model ones, which correspond instead to \emph{negative} rational~$b^2$.

\bibliographystyle{plain}
\bibliography{Radius2star}

@article{1404.7075,
    author = "Bershtein, M. and Foda, O.",
    title = "{AGT, Burge pairs and minimal models}",
    eprint = "1404.7075",
    archivePrefix = "arXiv",
    primaryClass = "hep-th",
    doi = "10.1007/JHEP06(2014)177",
    journal = "JHEP",
    volume = "06",
    pages = "177",
    year = "2014"
}

@article{1404.7094,
    author = "Alkalaev, K. B. and Belavin, V. A.",
    title = "{Conformal blocks of $W_N$ minimal models and AGT correspondence}",
    eprint = "1404.7094",
    archivePrefix = "arXiv",
    primaryClass = "hep-th",
    reportNumber = "FIAN-TD-2014-08",
    doi = "10.1007/JHEP07(2014)024",
    journal = "JHEP",
    volume = "07",
    pages = "024",
    year = "2014"
}

@article{1507.03540,
    author = "Belavin, Vladimir and Foda, Omar and Santachiara, Raoul",
    title = "{AGT, N-Burge partitions and $ {\mathcal{W}}_N $ minimal models}",
    eprint = "1507.03540",
    archivePrefix = "arXiv",
    primaryClass = "hep-th",
    doi = "10.1007/JHEP10(2015)073",
    journal = "JHEP",
    volume = "10",
    pages = "073",
    year = "2015"
}

@article{1909.10784,
    author = "Ribault, Sylvain",
    title = "{The non-rational limit of D-series minimal models}",
    eprint = "1909.10784",
    archivePrefix = "arXiv",
    primaryClass = "hep-th",
    doi = "10.21468/SciPostPhysCore.3.1.002",
    journal = "SciPost Phys. Core",
    volume = "3",
    pages = "002",
    year = "2020"
}

@article{1711.08916,
    author = "Migliaccio, Santiago and Ribault, Sylvain",
    title = "{The analytic bootstrap equations of non-diagonal two-dimensional CFT}",
    eprint = "1711.08916",
    archivePrefix = "arXiv",
    primaryClass = "hep-th",
    doi = "10.1007/JHEP05(2018)169",
    journal = "JHEP",
    volume = "05",
    pages = "169",
    year = "2018"
}

@misc{1406.4290,
    author = "Ribault, Sylvain",
    title = "{Conformal field theory on the plane}",
    eprint = "1406.4290",
    archivePrefix = "arXiv",
    primaryClass = "hep-th",
    month = "6",
    year = "2014"
}

@misc{2512.03172,
    author = "Roussillon, Julien and Tsiares, Ioannis",
    title = "{On the Virasoro Crossing Kernels at Rational Central Charge}",
    eprint = "2512.03172",
    archivePrefix = "arXiv",
    primaryClass = "hep-th",
    month = "12",
    year = "2025"
}

@inbook{1412.7121,
author="Tachikawa, Yuji",
unused-editor="Teschner, J{\"o}rg",
title="A Review on Instanton Counting and W-Algebras",
bookTitle="New Dualities of Supersymmetric Gauge Theories",
year="2016",
publisher="Springer International Publishing",
address="Cham",
pages="79--120",
isbn="978-3-319-18769-3",
doi="10.1007/978-3-319-18769-3_4",
url="https://doi.org/10.1007/978-3-319-18769-3_4",
unused-crossref="Teschner:2016yzf"
}

@book{Teschner:2016yzf,
    editor = {Teschner, J{\"o}rg},
    title = "{New Dualities of Supersymmetric Gauge Theories}",
    doi = "10.1007/978-3-319-18769-3",
    isbn = "978-3-319-18768-6, 978-3-319-18769-3",
    publisher = "Springer",
    address = "Cham, Switzerland",
    series = "Mathematical Physics Studies",
    year = "2016"
}

@misc{1504.01238,
      title={On convergence of basic hypergeometric series}, 
      author={Toshio Oshima},
      year={2015},
      eprint={1504.01238},
      archivePrefix={arXiv},
      primaryClass={math.CA},
      url={https://arxiv.org/abs/1504.01238}, 
}

@article{2205.00253,
  title={On the greatest common divisor of integer parts of polynomials},
  author={Banks, William and Shparlinski, Igor E},
  journal={Michigan Mathematical Journal},
  volume={75},
  number={3},
  pages={451--469},
  year={2025},
  publisher={University of Michigan, Department of Mathematics}
}

@article{hep-th/0206161,
    author = "Nekrasov, Nikita A.",
    title = "{Seiberg-Witten prepotential from instanton counting}",
    eprint = "hep-th/0206161",
    archivePrefix = "arXiv",
    reportNumber = "ITEP-TH-22-02, IHES-P-04-22",
    doi = "10.4310/ATMP.2003.v7.n5.a4",
    journal = "Adv. Theor. Math. Phys.",
    volume = "7",
    number = "5",
    pages = "831--864",
    year = "2003"
}

@article{hep-th/0701156,
    author = "Iqbal, Amer and Kozcaz, Can and Vafa, Cumrun",
    title = "{The Refined topological vertex}",
    eprint = "hep-th/0701156",
    archivePrefix = "arXiv",
    doi = "10.1088/1126-6708/2009/10/069",
    journal = "JHEP",
    volume = "10",
    pages = "069",
    year = "2009"
}

@article{0906.3219,
    author = "Alday, Luis F. and Gaiotto, Davide and Tachikawa, Yuji",
    title = "{Liouville Correlation Functions from Four-dimensional Gauge Theories}",
    eprint = "0906.3219",
    archivePrefix = "arXiv",
    primaryClass = "hep-th",
    doi = "10.1007/s11005-010-0369-5",
    journal = "Lett. Math. Phys.",
    volume = "91",
    pages = "167--197",
    year = "2010"
}

@inproceedings{0908.4052,
    author = "Nekrasov, Nikita A. and Shatashvili, Samson L.",
    title = "{Quantization of Integrable Systems and Four Dimensional Gauge Theories}",
    booktitle = "{16th International Congress on Mathematical Physics}",
    eprint = "0908.4052",
    archivePrefix = "arXiv",
    primaryClass = "hep-th",
    reportNumber = "TCD-MATH-09-19, HMI-09-09, IHES-P-09-38",
    doi = "10.1142/9789814304634_0015",
    pages = "265--289",
    year = "2010"
}

@article{1207.0787,
    author = "Gamayun, O. and Iorgov, N. and Lisovyy, O.",
    title = "{Conformal field theory of Painlev{\'e} VI}",
    eprint = "1207.0787",
    archivePrefix = "arXiv",
    primaryClass = "hep-th",
    doi = "10.1007/JHEP10(2012)038",
    journal = "JHEP",
    volume = "10",
    pages = "038",
    year = "2012",
    note = "[Erratum: JHEP 10, 183 (2012)]"
}

@article{1404.5188,
    author = "Novaes, F{\'a}bio and Carneiro da Cunha, Bruno",
    title = "{Isomonodromy, Painlev{\'e} transcendents and scattering off of black holes}",
    eprint = "1404.5188",
    archivePrefix = "arXiv",
    primaryClass = "hep-th",
    doi = "10.1007/JHEP07(2014)132",
    journal = "JHEP",
    volume = "07",
    pages = "132",
    year = "2014"
}

@article{2006.06111,
    author = "Aminov, Gleb and Grassi, Alba and Hatsuda, Yasuyuki",
    title = "{Black Hole Quasinormal Modes and Seiberg{\textendash}Witten Theory}",
    eprint = "2006.06111",
    archivePrefix = "arXiv",
    primaryClass = "hep-th",
    reportNumber = "RUP-20-18",
    doi = "10.1007/s00023-021-01137-x",
    journal = "Annales Henri Poincare",
    volume = "23",
    number = "6",
    pages = "1951--1977",
    year = "2022"
}

@article{math/0609841,
    author = "Martens, Johan",
    title = "{Equivariant volumes of non-compact quotients and instanton counting}",
    eprint = "math/0609841",
    archivePrefix = "arXiv",
    doi = "10.1007/s00220-008-0501-x",
    journal = "Commun. Math. Phys.",
    volume = "281",
    pages = "827--857",
    year = "2008"
}

@article{1403.1235,
    author = "Its, Alexander and Lisovyy, Oleg and Tykhyy, Yuriy",
    title = "{Connection Problem for the Sine-Gordon/Painlev{\'e} III Tau Function and Irregular Conformal Blocks}",
    eprint = "1403.1235",
    archivePrefix = "arXiv",
    primaryClass = "math-ph",
    doi = "10.1093/imrn/rnu209",
    journal = "Int. Math. Res. Not.",
    volume = "2015",
    number = "18",
    pages = "8903--8924",
    year = "2015"
}

@article{1608.02566,
    author = "Bershtein, M. A. and Shchechkin, A. I.",
    title = "{q-deformed Painlev{\'e} $\tau$ function and q-deformed conformal blocks}",
    eprint = "1608.02566",
    archivePrefix = "arXiv",
    primaryClass = "math-ph",
    doi = "10.1088/1751-8121/aa5572",
    journal = "J. Phys. A",
    volume = "50",
    number = "8",
    pages = "085202",
    year = "2017"
}

@article{2212.06741,
    author = "Arnaudo, Paolo and Bonelli, Giulio and Tanzini, Alessandro",
    title = "{On the Convergence of Nekrasov Functions}",
    eprint = "2212.06741",
    archivePrefix = "arXiv",
    primaryClass = "hep-th",
    doi = "10.1007/s00023-023-01349-3",
    journal = "Annales Henri Poincare",
    volume = "25",
    number = "4",
    pages = "2389--2425",
    year = "2024"
}

@article{1709.05232,
    author = {Felder, Giovanni and M{\"u}ller-Lennert, Martin},
    title = "{Analyticity of Nekrasov Partition Functions}",
    eprint = "1709.05232",
    archivePrefix = "arXiv",
    primaryClass = "math-ph",
    doi = "10.1007/s00220-018-3270-1",
    journal = "Commun. Math. Phys.",
    volume = "364",
    number = "2",
    pages = "683--718",
    year = "2018"
}

@article{2204.04409,
    author = "Huang, Yi-Zhi",
    title = "{Convergence in Conformal Field Theory}",
    eprint = "2204.04409",
    archivePrefix = "arXiv",
    primaryClass = "math.QA",
    doi = "10.1007/s11401-022-0379-5",
    journal = "Chin. Ann. Math. B",
    volume = "43",
    number = "6",
    pages = "1101--1124",
    year = "2022"
}

@article{2003.03802,
    author = "Ghosal, Promit and Remy, Guillaume and Sun, Xin and Sun, Yi",
    title = "{Probabilistic conformal blocks for Liouville CFT on the torus}",
    eprint = "2003.03802",
    archivePrefix = "arXiv",
    primaryClass = "math.PR",
    doi = "10.1215/00127094-2023-0031",
    journal = "Duke Math. J.",
    volume = "173",
    number = "6",
    pages = "1085--1175",
    year = "2024"
}

@article{2405.09325,
    author = "Roussillon, Julien",
    title = "{On the Virasoro fusion and modular kernels at any irrational central charge}",
    eprint = "2405.09325",
    archivePrefix = "arXiv",
    primaryClass = "hep-th",
    doi = "10.21468/SciPostPhys.17.5.138",
    journal = "SciPost Phys.",
    volume = "17",
    number = "5",
    pages = "138",
    year = "2024"
}

@article{2006.14025,
    author = "Le Floch, Bruno",
    title = "{A slow review of the AGT correspondence}",
    eprint = "2006.14025",
    archivePrefix = "arXiv",
    primaryClass = "hep-th",
    doi = "10.1088/1751-8121/ac5945",
    journal = "J. Phys. A",
    volume = "55",
    number = "35",
    pages = "353002",
    year = "2022"
}

@article{1504.01925,
    author = "Foda, Omar and Wu, Jian-Feng",
    title = "{From topological strings to minimal models}",
    eprint = "1504.01925",
    archivePrefix = "arXiv",
    primaryClass = "hep-th",
    doi = "10.1007/JHEP07(2015)136",
    journal = "JHEP",
    volume = "07",
    pages = "136",
    year = "2015"
}

@article{1509.01000,
    author = "Fukuda, Masayuki and Nakamura, Satoshi and Matsuo, Yutaka and Zhu, Rui-Dong",
    title = "{SH$^{c}$ realization of minimal model CFT: triality, poset and Burge condition}",
    eprint = "1509.01000",
    archivePrefix = "arXiv",
    primaryClass = "hep-th",
    reportNumber = "UT-15-32",
    doi = "10.1007/JHEP11(2015)168",
    journal = "JHEP",
    volume = "11",
    pages = "168",
    year = "2015"
}

@article{2011.06292,
    author = "Del Monte, Fabrizio and Desiraju, Harini and Gavrylenko, Pavlo",
    title = "{Isomonodromic Tau Functions on a Torus as Fredholm Determinants, and Charged Partitions}",
    eprint = "2011.06292",
    archivePrefix = "arXiv",
    primaryClass = "math-ph",
    doi = "10.1007/s00220-022-04458-y",
    journal = "Commun. Math. Phys.",
    volume = "398",
    number = "3",
    pages = "1029--1084",
    year = "2023"
}

@article{1608.00958,
    author = "Gavrylenko, P. and Lisovyy, O.",
    title = "{Fredholm Determinant and Nekrasov Sum Representations of Isomonodromic Tau Functions}",
    eprint = "1608.00958",
    archivePrefix = "arXiv",
    primaryClass = "math-ph",
    doi = "10.1007/s00220-018-3224-7",
    journal = "Commun. Math. Phys.",
    volume = "363",
    pages = "1--58",
    year = "2018"
}

@article{1705.01869,
    author = "Gavrylenko, P. and Lisovyy, O.",
    editor = "Kashani-Poor, Amir-Kian and Minasian, Ruben and Nekrasov, Nikita and Pioline, Boris",
    title = "{Pure SU(2) gauge theory partition function and generalized Bessel kernel.}",
    eprint = "1705.01869",
    archivePrefix = "arXiv",
    primaryClass = "math-ph",
    doi = "10.1090/pspum/098/01727",
    journal = "Proc. Symp. Pure Math.",
    volume = "18",
    pages = "181--208",
    year = "2018"
}

@article{1712.08546,
    author = "Cafasso, M. and Gavrylenko, P. and Lisovyy, O.",
    title = "{Tau functions as Widom constants}",
    eprint = "1712.08546",
    archivePrefix = "arXiv",
    primaryClass = "math-ph",
    doi = "10.1007/s00220-018-3230-9",
    journal = "Commun. Math. Phys.",
    volume = "365",
    number = "2",
    pages = "741--772",
    year = "2019"
}

@article{2508.14030,
    author = "Del Monte, Fabrizio and Desiraju, Harini and Gavrylenko, Pavlo",
    title = "{Modular transformations of tau functions and conformal blocks on the torus}",
    eprint = "2508.14030",
    archivePrefix = "arXiv",
    primaryClass = "math-ph",
    month = "8",
    year = "2025"
}

\end{document}